\documentclass[leqno]{article}
\usepackage{ltexpprt}
\usepackage{geometry}
\usepackage[english]{babel} 
\usepackage{graphicx} 
\usepackage{amsmath}
\usepackage{amssymb}
\usepackage{xcolor}
\usepackage{xspace}
\usepackage{gensymb}
\usepackage{subfigure}
\usepackage{paralist}
\usepackage[textsize=small]{todonotes}
\usepackage{hyperref}

\newenvironment{claimx}{\medskip\noindent\textbf{Claim.~}}{}

\newcommand{\remove}[1]{}
\newtheorem{observation}{Observation}

\newcommand{\NP}{$\mathcal{NP}$\xspace}

\renewcommand{\NP}{$\mathcal{NP}$}
\newcommand{\NPC}{\mbox{\NP-complete}\xspace}

\newcommand{\NPH}{\mbox{\NP-hard}\xspace}
\newcommand{\NorthEast}{NE\xspace}
\newcommand{\NorthWest}{NW\xspace}
\newcommand{\SouthEast}{SE\xspace}
\newcommand{\SouthWest}{SW\xspace}

\definecolor{blue}{rgb}{0.274,0.392,0.666}
\definecolor{red}{rgb}{0.627,0.117,0.156}
\definecolor{green}{rgb}{0,0.588,0.509}

\newcommand{\qc}{q-con\-strained\xspace}
\newcommand{\qconstrained}{quadrant-con\-strained\xspace}

\newcommand{\scalednwarrow}{\ensuremath{\text{\scalebox{0.5}{$\nwarrow$}}}}
\newcommand{\scalednearrow}{\ensuremath{\text{\scalebox{0.5}{$\nearrow$}}}}
\newcommand{\scaledswarrow}{\ensuremath{\text{\scalebox{0.5}{$\swarrow$}}}}
\newcommand{\scaledsearrow}{\ensuremath{\text{\scalebox{0.5}{$\searrow$}}}}

\renewcommand{\angle}[1]{\ensuremath{\langle #1 \rangle}\xspace}

\newcommand{\scaledsymbol}[1]{\ensuremath{\text{\scalebox{0.7}{$#1$}}}}
\newcommand{\Scalednwarrow}{\scaledsymbol{\nwarrow}}
\newcommand{\Scalednearrow}{\scaledsymbol{\nearrow}}
\newcommand{\Scaledswarrow}{\scaledsymbol{\swarrow}}
\newcommand{\Scaledsearrow}{\scaledsymbol{\searrow}}

\newcommand{\sets}{\{\ensuremath{\Scalednearrow, \Scalednwarrow, \Scaledswarrow, \Scaledsearrow}\}}

\def\Nur(#1){\ensuremath{\overset{\scalednearrow}{#1}}}
\def\Nul(#1){\ensuremath{\overset{\scalednwarrow}{#1}}}
\def\Ndr(#1){\ensuremath{\overset{\scaledsearrow}{#1}}}
\def\Ndl(#1){\ensuremath{\overset{\scaledswarrow}{#1}}}
\def\Npar(#1,#2){\ensuremath{\overset{\scaledsymbol{#2}}{#1}}}

\def\capNur(#1){\ensuremath{\overset{\protect\scalednearrow}{#1}}}
\def\capNul(#1){\ensuremath{\overset{\protect\scalednwarrow}{#1}}}
\def\capNdr(#1){\ensuremath{\overset{\protect\scaledsearrow}{#1}}}
\def\capNdl(#1){\ensuremath{\overset{\protect\scaledswarrow}{#1}}}
\def\capNpar(#1,#2){\ensuremath{\overset{\protect\scaledsymbol{#2}}{#1}}}

\newcommand{\ur}{\ensuremath{\Scalednearrow}}
\newcommand{\ul}{\ensuremath{\Scalednwarrow}}
\newcommand{\dr}{\ensuremath{\Scaledsearrow}}
\newcommand{\dl}{\ensuremath{\Scaledswarrow}}

\newcommand{\capur}{\protect\ur}

\newcommand{\WN}{\ensuremath{\mathrm{N}}}
\newcommand{\WE}{\ensuremath{\mathrm{E}}}
\newcommand{\WS}{\ensuremath{\mathrm{S}}}
\newcommand{\WW}{\ensuremath{\mathrm{W}}}

\begin{document}

\title{{Windrose Planarity:}\\
{Embedding Graphs with Direction-Constrained Edges}%
\thanks{This work was initiated at the Bertinoro
Workshop on Graph Drawing 2015, which was supported by the European Science
Foundation as part of the EuroGIGA collaborative research program (Graphs in
Geometry and Algorithms). A preliminary version of this work
was presented at the Symposium on Discrete Algorithms (SODA 2016) in
San Francisco~\cite{paddkrr-wp-16}.
Angelini was partially supported by DFG grant Ka812/17-1.
Da Lozzo and Di Battista were partially supported 
  by MIUR Project ``MODE'' under PRIN 20157EFM5C and 
  by H2020-MSCA-RISE project 734922 – ``CONNECT''.
This work was also supported in part by the MIUR-DAAD Joint Mobility Program: N$^\circ$ 34120 and N$^\circ$ 57397196.\xspace}
}

\author{Patrizio Angelini
\thanks{Wilhelm-Schickard-Institut f\"ur Informatik, Universit\"at T\"ubingen, Germany, \href{mailto:angelini@informatik.uni-tuebingen.de}{\em angelini@informatik.uni-tuebingen.de}.
}
\and 
Giordano {Da Lozzo}
\thanks{Department of Engineering, Roma Tre University, Italy. 
\href{mailto:dalozzo@dia.uniroma3.it}{\em dalozzo@dia.uniroma3.it}, 
\href{mailto:gdb@dia.uniroma3.it}{\em gdb@dia.uniroma3.it}, and
\href{mailto:didonato@dia.uniroma3.it}{\em didonato@dia.uniroma3.it}.
}
\and
Giuseppe {Di Battista}
\footnotemark[3]
\and
Valentino {Di Donato}
\footnotemark[3]
\and
Philipp Kindermann
\thanks{David R. Cheriton School of Computer Science, University of Waterloo, Canada.
\href{mailto:pkinderm@uwaterloo.ca}{\em pkinderm@uwaterloo.ca}.
}
\and
G\"unter Rote\thanks{Institut f\"ur Informatik, Freie Universit\"at Berlin, Germany.
\href{mailto:rote@inf.fu-berlin.de}{\em rote@inf.fu-berlin.de}.
}
\and
Ignaz Rutter
\thanks{Department of Computer Science and Mathematics, University of Passau, Germany.
\href{mailto:rutter@fim.uni-passau.de}{\em rutter@fim.uni-passau.de}.
}
}

\date{}
\maketitle

\begin{abstract} 
      Given a planar graph $G$ and a partition of the neighbors of each vertex $v$ in four sets $\Nur(v)$, $\Nul(v)$, $\Ndl(v)$, and $\Ndr(v)$, the problem {\sc Windrose Planarity} asks to decide whether $G$ admits a \emph{windrose-planar drawing}, that is, a planar drawing in which 
      \begin{inparaenum}[(i)]
      \item each neighbor $u \in \Nur(v)$ is above and to the right of $v$, 
      \item each neighbor $u \in \Nul(v)$ is above and to the left of $v$, 
      \item each neighbor $u \in \Ndl(v)$ is below and to the left of $v$, 
      \item each neighbor $u \in \Ndr(v)$ is below and to the right of $v$, and
      \item edges are represented by curves that are monotone with respect to each axis.
      \end{inparaenum}
      By exploiting both the horizontal and the vertical relationship among vertices, windrose-planar drawings allow to simultaneously visualize two partial orders defined by means of the edges of the graph.

      Although the problem is \NPH in the general case, we give a
      polynomial-time algorithm for testing whether there exists a
      windrose-planar drawing that respects a given combinatorial embedding.
This algorithm is based on a
      characterization of the plane triangulations admitting a
      windrose-planar drawing. Furthermore, for any embedded graph
      with $n$ vertices that has a windrose-planar drawing, we can construct one with at most one bend per edge and with at most $2n-5$ bends in total, which lies on the $3n \times 3n$ grid. The latter result contrasts with the fact that straight-line windrose-planar drawings may require exponential area.
\end{abstract}

\maketitle

\section{Introduction} \label{se:introduction}

Planarity is among the most studied topics in Graph Algorithms and Graph Theory. A great body of literature is devoted to the study of constrained notions of planarity. Classical examples are clustered planarity~\cite{addfpr-rccp-15,cdfk-atcpefcg-14,FengCE95}, in which vertices are constrained into prescribed regions of the plane called clusters, 
level planarity~\cite{addfr-tibp-tcs-15,jlm-lptlt-98}, in which vertices are assigned to horizontal lines,
strip planarity~\cite{addf-gd-13}, in which vertices have to lie inside parallel strips of the plane, and upward planarity. 
A directed acyclic graph is \emph{upward-planar} if it admits a planar drawing in which, for each directed edge $(u,v)$, vertex $u$ lies below $v$ and $(u,v)$ is represented by a $y$-monotone curve. 
Intuitively, edges ``flow'' from South to North. 
While testing upward planarity is in general \NPH~\cite{gt-ccurp-01}, the case in which a combinatorial embedding of the graph is prescribed can be tested in polynomial time~\cite{bdlm-udtg-94}.

\begin{figure}[t]
      \centering
      \subfigure[]{\includegraphics[page=1]{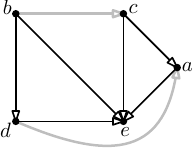}\label{fig:example-1}}
      \hfil
      \subfigure[]{\includegraphics[page=2]{img/example}\label{fig:example-2}}
      \caption{(a) A graph with specified quadrant directions for each edge (black edges are directed bottom-right, gray edges are directed top-right),
      (b) a windrose-planar drawing.}
      \label{fig:example}
\end{figure}

We introduce and study \emph{windrose planarity}, a notion of planarity that naturally generalizes upward planarity.
While upward 
planarity considers only two cardinal points (North and South), windrose planarity 
distinguishes four directions (\NorthEast, \NorthWest, \SouthWest, and 
\SouthEast), which we denote by the symbols $\Scalednearrow, \Scalednwarrow, 
\Scaledswarrow$, and~$\Scaledsearrow$, respectively.
More formally, let $G$ be a graph and suppose that, for each vertex~$v$ of~$G$,
its neighbors are partitioned into four (possibly empty) sets $\Nur(v)$, 
$\Nul(v)$, $\Ndl(v)$, and~$\Ndr(v)$. A drawing of $G$ is \emph{windrose-planar} 
if it is planar, edges are drawn as curves that are both $x$- and $y$-monotone, and, for each vertex $v$, the vertices in $\Nur(v)$, $\Nul(v)$, $
\Ndl(v)$, and~$\Ndr(v)$ lie \NorthEast, \NorthWest, \SouthWest, and \SouthEast of $v$, respectively; see Fig.~\ref{fig:example} for an illustration. The problem {\sc Windrose Planarity} takes as input a graph $G$ with a partition of its vertices into sets $\Nur(v)$, $\Nul(v)$, $\Ndl(v)$, and~$\Ndr(v)$, and asks to test the existence of a windrose-planar drawing.
Note that, in order to obtain a windrose-planar drawing of the graph in Fig.~\ref{fig:example}, the embedding has to be changed. However, even windrose-planar drawings with the same embedding might be very different; see Fig.~\ref{fig:windrose-ambiguous} for two windrose-planar drawings of the same cycle.
While upward-planar drawings represent \emph{one} partial order, windrose-planar drawings can be used to represent \emph{two} independent partial orders, given through the same edges.

The study of this problem is also motivated by its close relation with
{\sc Bi-Monotonicity}, introduced by Fulek, Pelsmajer, Schaefer and
\v{S}tefankovi\v{c}~\cite{fpss-htmd-11} and recently proved \NPC by Klemz and 
Rote~\cite{kr-olp-17}. This problem is similar to {\sc Windrose Planarity}, but 
the relative positions are specified for all pairs of vertices, not only for 
adjacent vertices.  A similar notion for grid drawings has been
studied previously under the name of Manhattan-geodesic
drawings~\cite{dgklr-hoap-12,kkrw-mgepg-09}.
Di Giacomo, Didimo, Kaufmann, Liotta, and Montecchiani~\cite{ddklm-urpd-14} 
studied the problem of constructing a planar
drawing of a directed graph where each edge can be drawn either
$x$-monotone or $y$-monotone. While {\sc Windrose Planarity} is a
generalization of upward planarity, this problem is a relaxation of
it. Di Battista,
Kim, Liotta, Lubiw, and Whitesides~\cite{dkllw-tsoctd-12} investigated the problem of computing a crossing-free 3D orthogonal drawing of a cycle whose edges have each been assigned a desired direction (East, West, North, South, Up, or Down).

Our contributions are as follows. Given that {\sc Windrose Planarity}
generalizes {\sc Upward Planarity}, we get as an immediate consequence that
testing {\sc Windrose Planarity} is \NPH
and that straight-line windrose-planar drawings
may require exponential area (Section~\ref{se:upward}). Hence, we study plane graphs, that is, planar graphs with a fixed combinatorial embedding.
Our main contribution is to provide a polynomial-time testing algorithm for this case (Sections~\ref{se:triangulated} and~\ref{sec:necessary}). We remark that, since every triconnected planar graph has a unique combinatorial embedding (up to a choice of the outer face), our result implies a polynomial-time testing algorithm also for these graphs.

The algorithm is based on the following main concepts. Let $G$ be a plane graph. First, we introduce (Section~\ref{sec:angle-categories}) a \emph{labeling} with angle categories $0^\circ$, $90^\circ$, $180^\circ$, $270^\circ$, and $360^\circ$ of the pairs of edges that are consecutive on a face of $G$; a labeling is \emph{angular} if labels around a vertex sum up to $360^\circ$ and for each internal (external) facial cycle of length $k$ the sum of the labels respects the formula for the angle sum of a $k$-gon, i.e., $k \cdot 180^\circ - 360^\circ$ ($k \cdot 180^\circ + 360^\circ$).
Second, we show that if $G$ is a triangulation, then the constraints on the relative positions of the adjacent vertices of $G$ naturally determine a unique labeling. We prove that if such a labeling is angular, then $G$ is windrose-planar and admits a $1$-bend windrose-planar drawing on the $3n \times 3n$ grid.
The proof is based on an augmentation technique that transforms the triangulation in such a way that its angular labeling only has $0^\circ$ and $90^\circ$ labels.
Third, we deal with general plane graphs. In this case, it is no
longer true that the constraints on the relative positions of the
adjacent vertices of~$G$ determine a unique labeling. We show how to
find an angular labeling if it exists, by solving a flow problem in a
planar network.
Fourth, we show that a plane graph with an angular labeling can be augmented to a triangulation with an angular labeling. Finally, we directly get a characterization of the plane graphs admitting a windrose-planar drawing based on the above arguments.

We also investigate the question whether
a windrose-planar graph admits a straight-line windrose-planar
drawing; observe that this is always true for upward-planar
graphs~\cite{dt-aprad-88}. Even though we do not answer this question in
its entirety, we present three interesting related results.
 First, we give an algorithm to construct 
 windrose-planar drawings of windrose-planar graphs with at most one bend per edge, and whose vertices and bends lie on a polynomial-size grid
 (Section~\ref{se:triangulated}); we remark that straight-line drawings of windrose-planar graphs may require exponential area (Section~\ref{se:preliminaries} and~\ref{se:straight-line}).
Second, we provide an algorithm to compute straight-line windrose-planar drawings for a notable class of graphs (Section~\ref{se:straight-line}). Third, we answer in the negative an open question by Fulek et al.~\cite{fpss-htmd-11} about the straight-line realizability of bi-monotone drawings (Section~\ref{se:straight-line}). 

\section{Preliminaries} \label{se:preliminaries}

A planar \emph{drawing} of a graph~$G$ maps the vertices of~$G$ to 
distinct points in the plane and the edges 
of~$G$ to simple interior-disjoint Jordan curves between their endpoints. 
A planar drawing~$\Gamma$ induces a \emph{combinatorial embedding}, or \emph{planar embedding}, which is the class of topologically equivalent drawings. In particular, an 
embedding specifies the regions of the plane, called \emph{faces}, whose boundary 
consists of a cyclic sequence of edges. The unbounded face is called the 
\emph{outer face}, the other faces are called \emph{internal faces}.  Vertices and edges incident to the outer face are called \emph{external} vertices and edges, respectively.  The other vertices and edges are \emph{internal}.
A \emph{plane graph} is a graph together with a combinatorial embedding, which also prescribes an outer face.
A vertex of a plane graph $G$ is \emph{external} if it is incident to the outer face of $G$ and \emph{internal} otherwise. A plane graph $G$ is \emph{(internally) triangulated} if all (internal) faces are $3$-cycles. A \emph{triangulation} is a triangulated plane graph.
In this paper, we only consider \emph{simple} graphs, that is, graphs without self-loops or multi-edges.

\begin{figure}[t]
\centering
\includegraphics[page=1]{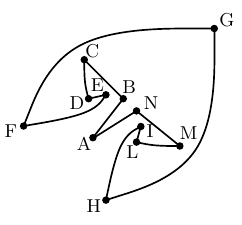}\label{fig:windrose-ambiguous-a}
\hfil
\includegraphics[page=2]{img/windrose-ambiguous}\label{fig:windrose-ambiguous-b}
\caption{Two windrose-planar drawings of the same cycle.  Note that for each vertex the set of neighbors in each of the four quadrants is the same in the two drawings.}
\label{fig:windrose-ambiguous}
\end{figure}

A \emph{separating $k$-set} of a graph is a set of $k$ vertices whose removal
increases the number of connected components.  A vertex constituting a
separating $1$-set is also called \emph{cutvertex}.  A graph is
$k$-connected if it has at least $k+1$ vertices and does not have any separating $(k-1)$-set.  Graphs
that are $2$-connected and $3$-connected are also called \emph{biconnected} and
\emph{triconnected}, respectively.  A \emph{block} of a graph $G$ is a maximal
subgraph that is biconnected.

A \emph{$k$-subdivision} $G'$ of a (plane) graph $G$ is a (plane)
graph obtained from $G$ by replacing each edge~$(u,v)$ with a path
between $u$ and $v$ containing at most $k$ intermediate vertices.
If the path replacing $(u,v)$ has at least one intermediate vertex, we call the 
edges of the path \emph{subdivision edges} and its intermediate vertices 
\emph{subdivision vertices}. An \emph{augmentation} $G'$ of a (plane) graph $G$ 
is a (plane) graph obtained from $G$ by adding some dummy vertices and edges.

\emph{Planar $3$-trees}, also known as \emph{stacked triangulations} or \emph{Apollonian graphs}, are special types of planar triangulations which can be generated from a triangle by a sequential addition of vertices of degree $3$ inside faces. 

A directed graph is \emph{acyclic} if it does not contain any directed cycle. 

An $x$-\emph{monotone} ($y$-\emph{monotone}) curve  is 
a curve that intersects every vertical line (every horizontal line) in at most one point.
An $xy$-\emph{monotone} curve is both $x$-monotone and $y$-monotone.

We say that a drawing is a \emph{$k$-bend drawing on a grid} if each
edge is represented by a polyline with at most $k$ bends, and all
vertices and bends of the edges lie on the points of an integer
grid. 
The \emph{area} of a drawing is the area of its bounding box, assuming that the minimum distance between any vertex or bend point and any other vertex or bend point is at least~$1$. In particular, since in any grid drawing vertices and bends lie on the integer grid, the area of a grid drawing
corresponds to the number of grid cells in its bounding box. 

The \emph{polar angle} of a vector~$\vec v$ is the clockwise angle from~$\vec v$ 
to the $x$-axis $(1,0)^\top$. The polar angle of a line segment~$s$ at one of
its end points~$u$ is the polar angle of the vector defined by~$s$ starting at~$u$.
In a polyline drawing of a graph, the polar angle of an edge~$e$ at a vertex~$v$ 
is the polar angle of the line segment of~$e$ incident to~$v$. The
\emph{geometric angle} between two edges~$e$ and~$e'$ at~$v$ is the 
counterclockwise angle between the line segment of~$e$ incident to~$v$ and the line segment of~$e'$
incident to~$v$ at~$v$.

\subsection{Windrose planarity}

A \emph{\qconstrained graph} (\emph{\qc graph}) is a pair $(G,Q)$
where $G$ is an undirected planar graph and $Q$ contains \emph{q-constraints}, that is,
a partition of the neighbors of each vertex $v$ into four sets $\Nur(v)$, $\Nul(v)$, $\Ndl(v)$, and~$\Ndr(v)$.  
We define the $\ur$-, $\ul$-, $\dl$-, and $\dr$-\emph{quadrant} of a
vertex $v$ in a drawing as the first, second, third, and fourth
quadrant around the point where $v$ lies, respectively.
A \emph{windrose-planar} drawing of a \qc graph $(G,Q)$ is a planar drawing of~$G$ such
that each edge $(u,v)$ is drawn as an $xy$-monotone curve and~$u$
lies in the $\circ$-quadrant of $v$, if $u \in \Npar (v,\circ)$.
We say that $(G,Q)$ is a \emph{windrose-planar} graph if it admits a windrose-planar drawing.
We assume throughout that the relative assignment of adjacent vertices
is \emph{consistent} in the following sense: for each edge $(u,v)$, we have $v
\in \Nul(u) \Leftrightarrow u \in \Ndr(v)$, and $v \in \Nur(u)
\Leftrightarrow  u \in \Ndl(v)$.
Whenever we add an edge $(u,v)$ to a \qc graph and assign $v$
to one of the four quadrants of $u$, we will implicitly assume that $u$ is added
to the appropriate quadrant of $v$ to maintain consistency.
The problem {\sc Windrose Planarity} asks whether a given \qc graph is windrose-planar; see Fig.~\ref{fig:bi-monotone-a} for a negative example.
Clearly, a disconnected \qc graph is windrose-planar if and only if all its connected components are. Hence, in the remainder of the paper we assume the graphs to be connected.

\begin{observation}\label{obs:planarity}
In any embedding corresponding to a windrose-planar drawing of a \qc
graph $(G,Q)$, for each vertex~$u$ of~$G$, the neighbors of~$u$ appear around~$u$ in this clockwise order starting at $0^\circ$:  first the vertices in~$\Ndr(u)$, then the vertices
in~$\Ndl(u)$, then the vertices in~$\Nul(u)$, and then the vertices 
in~$\Nur(u)$.
\end{observation}

Let $(G,Q)$ be a \qc plane graph with planar embedding $\mathcal{E}$.
The \emph{leftmost} (the \emph{rightmost}) \emph{neighbor} of a vertex $v$ of $G$ 
in $\Npar(v,\circ)$, with $\circ \in \sets$, is the neighbor 
$u \in \Npar(v,\circ)$, such that there exists no vertex 
$u' \in \Npar(v,\circ)$ that precedes (follows) $u$ in the clockwise order 
described in Observation~\ref{obs:planarity} of 
the neighbors around $v$ in $\mathcal{E}$. 
Note that such a neighbor might not exist.
We denote by~$G^\uparrow$ the graph obtained from $G$ by directing its 
edges from~$u$ to~$v$ if $v \in \Nul(u) \cup \Nur(u)$. Similarly, we denote
by~$G^\rightarrow$ the graph obtained from~$G$ by directing its edges 
from~$u$ to~$v$ if $v \in \Nur(u) \cup \Ndr(u)$. 

\subsection{Relationship with Upward Planarity}\label{se:upward}

{\sc Windrose Planarity} has a close relationship with {\sc Upward Planarity}~\cite{bdlm-udtg-94}, which is defined as follows. Let $D$ be a directed graph. An \emph{upward-planar} drawing of $D$ is a planar drawing in which each directed edge $(u,v)$ is drawn as a $y$-monotone curve such that vertex~$u$ lies below vertex $v$. The {\sc Upward Planarity} problem asks whether $D$ admits an upward-planar drawing.
The first relationship concerns graphs $G^\uparrow$ and $G^\rightarrow$. We say that a \qc graph $(G,Q)$ is \emph{bi-acyclic} if $G^\uparrow$ and $G^\rightarrow$ are acyclic. 
Recall that upward-planar graphs are acyclic, hence any \qc windrose-planar graph is bi-acyclic. It is easy to see that the upward planarity of $G^\uparrow$ and $G^\rightarrow$ is a necessary condition for the windrose planarity of $(G,Q)$.
On the other hand, this condition is not sufficient, as shown in Fig.~\ref{fig:bi-monotone}.
\begin{figure}[t]
   \centering
   \subfigure[]{
   \includegraphics[page=1]{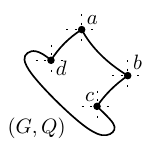}\label{fig:bi-monotone-a}
   }\hfil
   \subfigure[]{
   \includegraphics[page=2]{img/bi-monotone}\label{fig:bi-monotone-b}
   }\hfil
   \subfigure[]{
   \includegraphics[page=3]{img/bi-monotone}\label{fig:bi-monotone-c}
   }
   \caption{(a) Illustration of a \qc plane graph $(G,Q)$ that is not windrose-planar, indicating the desired quadrant for each edge.
  (b--c)~Upward-planar drawings of the plane graphs $G^\uparrow$ and $G^\rightarrow$.}
  \label{fig:bi-monotone}
\end{figure}
An even stronger relationship is that {\sc Windrose Planarity} is a generalization of {\sc Upward Planarity}. Namely, an instance of this latter problem can be translated into one of {\sc Windrose Planarity} by just placing
the outgoing neighbors and the incoming neighbors of a vertex $v$ into sets $\Nur(v)$ and $\Ndl(v)$, respectively. The two instances are then equivalent: If we have an upward-planar
drawing, we can assume that the edges are straight-line segments~\cite{dt-aprad-88}; then, we can make all slopes larger than $1$ in absolute value by scaling the $x$-axis; and finally perform a rotation by~$45^\circ$ to make the drawing windrose-planar. The other direction is trivial, as any windrose-planar drawing of $(G,Q)$ is also an upward-planar drawing of~$G^\uparrow$.
Since {\sc Upward Planarity} is \NPC~\cite{gt-ccurp-01} and since there are directed graphs requiring exponential area in any straight-line upward-planar drawing~\cite{dtt-arsdud-92}, the following negative results are immediate consequences.

\begin{theorem}
\label{th: np-complete}
  The problem {\sc Windrose Planarity} is \NPC.
\end{theorem}

\begin{theorem}
\label{th:area-lowerbound}
There exists an infinite family of \qc graphs $(G_n,Q_n)$ on~$n$ vertices such that any straight-line windrose-planar drawing of~$(G_n,Q_n)$ has area $\Omega(2^{n/2})$, under any \mbox{resolution rule.}
\end{theorem}

In Section~\ref{se:straight-line}, we strengthen Theorem~\ref{th:area-lowerbound} by exploiting all four sets of neighbors to give a \linebreak \qc $n$-vertex graph requiring $\Omega(4^{n/3})$ area in any straight-line windrose-planar drawing.

\section{Q-Constraints and angle categories}\label{sec:angle-categories}

In this section we develop an alternative angle-based description of \qc
graphs that is sometimes more useful than the quadrant-based view.

Consider a plane graph $G$. We call \emph{angle} an incidence between a vertex $v$ and a face $f$ in $G$. 
Every angle is bounded by two edges~$e$ and $e'$ incident to $v$, 
where $e$ precedes $e'$ in the clockwise circular order of the edges incident to $v$; 
we denote the angle by $\angle{e,e'}$ and say that $\angle{e,e'} \in f$. 
Since~$G$ is simple, if $e=e'$, then $v$ has degree~$1$. In this case, vertex $v$ has just one angle $\angle{e,e}$ that is bounded by the same edge from both sides. 
If a vertex $v$ has only one angle in $f$ (e.g., if the graph is $2$-connected),
then we also denote the angle by $\angle{v,f}$.

We consider five different \emph{angle categories}: $0^\circ$, $90^\circ$,
$180^\circ$, $270^\circ$, and $360^\circ$.  A \emph{labeling} $A$ assigns
to each angle of a plane graph one of these angle categories; see Fig.\ref{fig:angular-drawing-a}. 
A \emph{labeled graph} $(G,A)$ is a pair where $G$ is a plane graph and $A$ is a labeling of its angles.
An \emph{angular drawing} of $(G,A)$ is an $xy$-monotone drawing of $G$ such that 
\begin{inparaenum}[(i)]
\item every edge
starts and ends with a line segment whose slope is close to $1$
or~$-1$,
i.e., the polar angle 
 differs by at most $1^\circ$ from 
 $\pm 45^\circ$ or $\pm 135^\circ$, and
\item the geometric angle $\alpha$ between two consecutive 
 edges~$e$
  and~$e'$ incident to a vertex is close to the angle category of the corresponding angle
  $\angle{e,e'}$, i.e.,
  $|\alpha-A(\angle{e,e'})| < 2^\circ$.
\end{inparaenum}
See Fig.~\ref{fig:angular-drawing-b}.
Note that the choice of differing by~$1^\circ$ is arbitrary, we could define angular drawings
with any $\varepsilon^\circ$ for $0<\varepsilon<22.5$.

\begin{figure}[b!]
\centering
  \subfigure[]{
    \includegraphics[page=1]{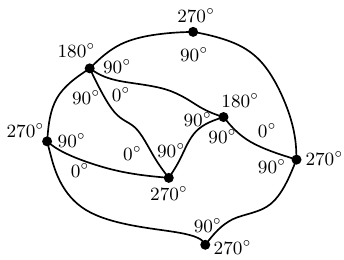}\label{fig:angular-drawing-a}
  }\hfil
  \subfigure[]{
    \includegraphics[page=2]{img/angular-drawing}\label{fig:angular-drawing-b}
  }
  \caption{(a) Windrose-planar drawing $\Gamma$ of a \qc plane graph $(G,Q)$. Angle categories of the labeling $A$ corresponding to~$\Gamma$ are shown. (b)~Angular drawing $\Gamma'$ of $(G,A)$; note that $\Gamma'$ is also a windrose-planar drawing of $(G,Q)$.} 
  \label{fig:angular-drawing}
\end{figure}

\newcommand{\lemmalabelingnecessary}{
Let $(G,A)$ be a labeled graph.  If $(G,A)$ admits an angular drawing, then $A$ satisfies the following conditions.
  \begin{compactenum}[(a)]
  \item \emph{Vertex Condition:} The sum of the incident angle categories is $360^\circ$ for every vertex.
  \item \emph{Cycle Condition:} For every internal face $f$ of length $k$, 
the sum of the
    angles is $k \cdot 180^\circ - 360^\circ$\textup; for
 the
    outer face $f$, the sum of its angle categories is $k \cdot 180^\circ +
    360^\circ$.
\end{compactenum}
}

\begin{lemma}
  \label{lem:angular-labeling-necessary}
  \lemmalabelingnecessary
\end{lemma}

\begin{proof}
  Let $\Gamma$ be an angular drawing of $(G,A)$, which we assume to be a polygonal drawing (i.e., each edge is a polyline) such that the two edge-segments incident to each bend lie inside two opposite quadrants with respect to the bend. This assumption is without loss of generality, since any angular drawing is $xy$-monotone and thus a polyline drawing with this property can be always constructed by adding a sufficiently large number of bends along the $xy$-monotone curves representing the edges.
  
  \begin{figure}[t]
    \centering
    \includegraphics[page=1]{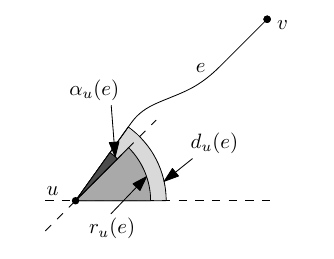}
    \caption{Illustrations for the definitions of direction, reference direction, and correction term.}
    \label{fig:correction-term}
  \end{figure}

  Let $e=(u,v)$ be an edge of $G$; refer to Fig.~\ref{fig:correction-term}. The \emph{direction}~$d_u(e)$ of $e$ at $u$ is the polar angle of $e$ at $u$, where $d_u(e) \in (-180^\circ,+180^\circ)$.
  The \emph{reference direction}~$r_u(e)$ of $e$ at $u$ is the polar angle of the diagonal line in the center of the angular range corresponding to the quadrant of $u$ containing $v$,  where $r_u(e) \in \{\pm 45^\circ, \pm 135^\circ\}$. 
  The \emph{correction term}~$\alpha_u(e)$ of $e$ at $u$ equals $d_u(e)-r_u(e)$. Note that $d_u(e)=180^\circ+d_v(e)$ and $r_u(e)=180^\circ+r_v(e)$, hence~$\alpha_u(e)=\alpha_v(e)$ and we simply write $\alpha(e)$.

  Consider an angle $\angle{e,e'}$ at a vertex $v$. We denote by $\beta(\angle{e,e'})$ the geometric angle of $e$ and $e'$ at $v$. By the definition of the correction terms we have the following: $\beta(\angle{e,e'}) = A(\angle{e,e'})-\alpha(e)+\alpha(e')$.
  
  Consider the edges $e_1,e_2,\dots,e_k$ in this circular order around vertex $v$, and let $e_{k+1}=e_{1}$. Then, $360^\circ=\sum_{i=1}^k{\beta(\angle{e_i,e_{i+1}})}=\sum_{i=1}^k{A(\angle{e_i,e_{i+1}})}$ since the correction terms cancel. Hence, the Vertex Condition holds. 

Consider the edges $e_1,e_2,\dots,e_k$ in this circular order along face $f$, and let $e_{k+1}=e_{1}$. We add a subdivision vertex at each bend of an edge incident to $f$. Note that this operation increases $k$; however, since
both the two edge-segments incident to each bend lie inside two opposite quadrants with respect to the bend, by assumption,
the internal angle of $f$ at each subdivision vertex has category $180^\circ$, and so the total effect neutralizes.  
Thus, we now have a $k$-gon (possibly with a larger $k$ than before) whose 
edges $e_i$ are straight-line segments. By the formula for the angle sum of a 
$k$-gon, we have $k \cdot 180^\circ - 360^\circ
=\sum_{i=1}^k{\beta(\angle{e_i,e_{i+1}})}
=\sum_{i=1}^k{A(\angle{e_i,e_{i+1}})}$, where the latter equality holds 
since the correction terms cancel out. The proof for the outer face is 
analogous. Hence, the Cycle Condition holds.
\end{proof}

A labeling satisfying the conditions of
Lemma~\ref{lem:angular-labeling-necessary} is called an \emph{angular labeling}.  
We now establish a connection between angular labelings and q-constraints.
Given a labeled graph $(G,A)$, the relative quadrants of the adjacent vertices with respect to each other in any angular drawing $\Gamma$ of $(G,A)$ are unique up to a rotation of $\Gamma$ by a multiple of $90^\circ$.
Hence, starting from~$A$, we can uniquely define q-constraints $Q_A$ for $G$ that preserve the circular order of their quadrants.
Namely, consider an angle $\angle{e,e'}$ bounded by edges $e=(v,u)$ and $e'=(v,w)$.
Assume, without loss of generality up to a rotation by a multiple of~$90^\circ$, that $u \in \Nur(v)$. 
If either $A(\angle{e,e'})=0^\circ$ or $A(\angle{e,e'})=360^\circ$, then $u$ and $w$ lie in the same quadrant of $v$ in any angular drawing of $G$ respecting~$A$; we represent this fact by setting $w \in \Nur(v)$; 
if $A(\angle{e,e'})=90^\circ$, then $w \in \Ndr(v)$; 
if $A(\angle{e,e'})=180^\circ$, then $w \in \Ndl(v)$; 
and if $A(\angle{e,e'})=270^\circ$, then $w \in \Nul(v)$. 
We formalize these concepts in the following observation.

\begin{observation}\label{obs:labeling-to-constraints}
An angular labeling $A$ defines a unique (up to cyclically shifting the quadrants) set $Q_A$ of q-constraints. Any angular drawing of~$(G,A)$ is a windrose-planar drawing of $(G,Q_A)$ (after a possible rotation by a multiple of $90^\circ$).
\end{observation}

Conversely, for a \qc graph $(G,Q)$ there may exist different angular 
labelings~$A$ and~$A'$ such that $Q_A=Q_{A'}=Q$. Suppose we are given a 
windrose-planar drawing~$\Gamma$ of a \qc graph $(G,Q)$.  Assume that the 
first and the last segment of each edge has slope close to~$1$ or to~$-1$; 
see Fig.~\ref{fig:angular-drawing-b}.  Now, all geometric angles at the 
vertices are  close to one of the angles in
 $\{0^\circ, 90^\circ, 180^\circ, 270^\circ, $ $360^\circ\}$ in~$\Gamma$.  This 
determines a unique angular labeling~$A_\Gamma$ of~$G$.  
However, the angular labeling $A_\Gamma$ depends on $\Gamma$, not only on $Q$. 
In fact, while angle categories of~$90^\circ, 180^\circ$, and~$270^\circ$ are 
uniquely defined by the q-constraints, the assignment of $0^\circ$ and $360^\circ$ 
angle categories is not unique. This is the case precisely for those vertices 
that have degree at least~$2$ and whose neighbors all lie in the same 
quadrant. We call such vertices \emph{ambiguous}. Vertices $A,C,E,F,G,H,I,M$ of Fig.~\ref{fig:windrose-ambiguous} are ambiguous, while the remaining vertices are not.
A \emph{large-angle assignment} $L$ assigns to each ambiguous vertex one of its incident angles.
The q-constraints $Q$ together with a large-angle assignment $L$ uniquely determine a labeling~$A_{Q,L}$~of~$G$.

\begin{observation}\label{obs:windrose-to-labeling}
Any windrose-planar drawing of $(G,Q)$ in which the first and the last segment of each edge has slope close to $1$ or $-1$ and the geometric angles close to $360^\circ$ comply with the large-angle assignment $L$ is an angular drawing of $(G,A_{Q,L})$.
\end{observation}

Once we have performed the large-angle assignment for each vertex, the 
problems of computing windrose-planar drawings and angular drawings are 
equivalent. Hence, in the following, we will refer to these notions of 
drawings interchangeably.

Let $G \subseteq G'$ be two plane graphs and let $A$ and $A'$ be
labelings of $G$ and $G'$, respectively.  We say that ~$A'$
\emph{refines} $A$ if for each angle $\angle{e,e'}$ of $G$ we have
$A(\angle{e,e'}) = \sum_{\angle{l,l'} \in C(e,e')} A(\angle{l,l'})$,
where~$C(e,e')$ denotes the angles of $G'$ clockwise between $e$ and $e'$.

\newcommand{\lemmamodifications}{Let $(G,A)$ be a labeled graph and let $G'$ be the graph obtained by adding to $G$ an edge $e$ inside a face $f$ of $G$. Also, let $f_1$ and $f_2$ be the two faces of $G'$ that are incident to $e$. Let $A'$ be a labeling of $G'$ refining $A$. Then, $A'$ is an angular labeling if and only if $A$ is an angular labeling and $f_1$ satisfies the Cycle Condition.}
\begin{lemma}
\label{lem:modifications}
\lemmamodifications
\end{lemma}

\begin{proof}
  Clearly, $A'$ satisfies the Vertex Condition for all vertices if and
  only if $A$ does, due to the fact that $A'$ refines $A$.
  We now consider the Cycle Condition. First observe that all faces of $G'$ not incident
  to $e$ also occur in $G$ with the same angles; hence their Cycle Conditions are equivalent.

  Suppose that $A'$ is an angular labeling, which implies that $f_1$ and $f_2$ satisfy the Cycle Condition. Hence, $s_i=\sum_{\angle{a,b}\in f_i} A'(\langle a,b\rangle) = 180^\circ \cdot k_i-360^\circ$, with $i=1,2$, where $k_i$ is the length of the facial cycle of~$f_i$. Since $A'$ refines $A$, we have that 
  \begin{eqnarray*}
  s &=& \sum_{\angle{a,b}\in f} A'(\langle a,b\rangle) = s_1 + s_2 \\
  &=& 180^\circ \cdot k_1-360^\circ + 180^\circ \cdot k_2-360^\circ \\
  &=& 180^\circ \cdot (k_1+k_2-2)-360^\circ.
  \end{eqnarray*}
  Given that the length of the facial cycle of $f$ is $k=k_1+k_2-2$, we have that $f$ satisfies the Cycle Condition.

  Suppose that $A$ is an angular labeling and that $f_1$ satisfies the Cycle Condition. Hence, $s=180^\circ \cdot k-360^\circ$ and $s_1=180^\circ \cdot k_1-360^\circ$. Since $A'$ refines $A$, we have that $s_2=s-s_1$. Thus, $s_2= 180^\circ \cdot k-360^\circ - (180^\circ \cdot k_1-360^\circ) = 180^\circ \cdot (k-k_1) = 180^\circ \cdot (k-k_1+2) - 360$. Given that $k_2=k-k_1+2$, we have that $f_2$ satisfies the Cycle Condition.
  This concludes the proof.
\end{proof}

\begin{corollary}\label{cor:subgraph-labeling}
Let $(G',A')$ be a labeled graph such that $A'$ is angular and let $G$ be a connected subgraph of $G'$ with the unique labeling $A$ such that $A'$ refines $A$. Then, $A$ is an angular labeling. In particular, for every simple cycle $C$ of $G'$, the face obtained by removing the interior of $C$ satisfies the Cycle Condition.
\end{corollary}

\section{Triangulated \qc graphs}\label{se:triangulated}

In this section, we consider windrose-planar drawings of
triangulated \qc graphs. We prove that, for these graphs,
the necessary condition that the labeling determined by the
q-constraints is angular (Lemma~\ref{lem:angular-labeling-necessary})
is also sufficient for windrose planarity. Further, we 
prove that a windrose-planar triangulated \qc graph admits a windrose-planar drawing with one bend per edge and polynomial area.

We start by observing some important properties of triangulated \qc graphs 
and of their angular labelings. The first observation directly follows from 
the Cycle Condition, which requires that the sum of the angle categories of 
the angles incident to each internal triangular face is $180^\circ$. 

\begin{observation}\label{obs:triangulated-angular}
Let $(G,Q)$ be a triangulated \qc graph with large-angle assignment $L$ and let $A_{Q,L}$ be the corresponding labeling of $G$. If $A_{Q,L}$ is angular, then no internal angle of $G$ has category larger than $180^\circ$ and each internal face $f$ of $G$ has at least an internal angle with category $0^\circ$.
\end{observation}

An implication of this observation is that no internal angle can be assigned category $360^\circ$. Thus, the only vertices for which a large-angle assignment may be needed are those incident to the outer face; even in this case, however, the assignment is fixed, since the $360^\circ$ category for these vertices must necessarily be assigned to the outer face. Hence, an angular labeling of a triangulated \qc graph, if any, is unique. In the remainder of the section, when considering a triangulated \qc graph, we will
thus omit explicitly referring to its large-angle assignment.

\newcommand{\lemmaangularnosource}{
  Let $(G,Q)$ be a triangulated \qc graph and suppose that the labeling $A_Q$ determined by $Q$ is angular. Then, 
\begin{inparaenum}[(i)]
\item the two graphs $G^\uparrow$ and $G^\rightarrow$ do not have internal sources or sinks, and
\item $(G,Q)$ is bi-acyclic.
\end{inparaenum}
}
\begin{lemma}\label{le:angular-no-source}
\lemmaangularnosource
\end{lemma}

\begin{figure}[tb]
  \centering
    \includegraphics[page=1]{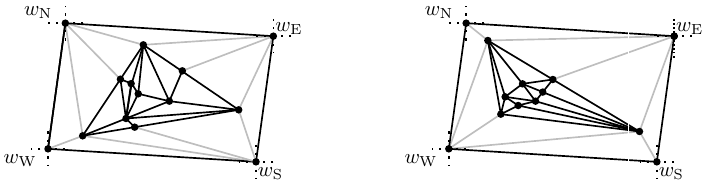}
\caption{Windrose-planar drawings of two quasi-triangulated \qc graph.}
\label{fig:quasi-triangulated}
\end{figure}

\begin{proof}
We first prove that $G^\uparrow$ and $G^\rightarrow$ do not have internal sources or sinks.
Assume that there is an internal source $v$ in $G^\uparrow$.
This implies that $\Ndr(v)=\Ndl(v)=\emptyset$, and hence $v$ has an
internal angle with category $270^\circ$ or $360^\circ$, which
contradicts Observation~\ref{obs:triangulated-angular}.
The cases
that $v$ is a sink or that $v$ is a source/sink in $G^\rightarrow$ are
analogous. 
 Thus, there are no internal sources and sinks.
  
To prove that $(G,Q)$ is bi-acyclic, we first show that the internal faces of $G^\uparrow$ and $G^\rightarrow$ are acyclic. Consider any internal face $f=(u,v,w)$ and the angle $\angle{e,e'}$, where $e=(v,u)$ and $e'=(v,w)$, such that $A(\angle{e,e'})=0^\circ$, which exists by Observation~\ref{obs:triangulated-angular}. By construction of $A_Q$, we have that $u,w \in \Npar(v,\circ)$, for some $\circ \in \sets$, and hence $e$ and $e'$ are either both outgoing or both incoming edges of $v$ in $G^\uparrow$ and $G^\rightarrow$. Hence, $f$ is acyclic in both $G^\uparrow$ and $G^\rightarrow$.

In order to prove that
 both $G^\uparrow$ and $G^\rightarrow$ are acyclic,
 we show by a counting argument that the acyclicity of each internal face of a triangulated directed planar graph $H$ without internal sources and sinks implies the acyclicity of the whole graph.
Namely, since each internal triangular face of $H$ is acyclic, it has a source and a sink vertex, plus a vertex that we call \emph{transition vertex}.
Assume, for a contradiction, that $H$ has a simple directed cycle $C$ with $n_C$ vertices.  
Consider the internally triangulated subgraph~$H'$ of~$H$ consisting of $C$ 
together with all vertices and edges in the interior of $C$. If $H'$ has $n_I$
interior vertices, then it has $2n_I+n_C-2$ internal triangular faces.
Since $H$ does not have internal sources or sinks, neither does~$H'$. 
Therefore,
every internal vertex~$v$ of~$H'$ is a transition vertex for at least
two faces, namely the faces whose boundary contains an incoming
and an outgoing edge of $v$. Also, every vertex $v$
on the cycle $C$ 
 is a transition vertex for at least one face
since, by assumption, $C$ is a directed cycle and hence $v$ has at
least one incoming and at least one outgoing edge in $H'$. Thus, there
are at least $2n_I+n_C$ pairs $v,f$ such that $v$ is a transition vertex for face $f$. However, there are only $2n_I+n_C-2$ faces, each with one transition vertex, a contradiction. Hence, $H$ is acyclic.
This concludes the proof of the lemma. 
\end{proof}

We now show
that a triangulated \qc graph $(G,Q)$ whose corresponding labeling $A_Q$ is angular is windrose-planar. By Observation~\ref{obs:triangulated-angular}, each internal angle of $(G,Q)$ has category $0^\circ$, $90^\circ$, or $180^\circ$ in $A_Q$. In Lemma~\ref{le:4-constrained-drawing}, we prove that, if no internal
angle has category $180^\circ$, then $(G,Q)$ admits a straight-line
windrose-planar drawing on the $n \times n$ grid. 
We prove this lemma under the assumption that $(G,Q)$ is internally triangulated and its outer face has four incident
vertices $w_{\WN}$, $w_{\WW}$, $w_{\WS}$, $w_{\WE}$, with $w_{\WN} \in
\Nur(w_{\WW})$, $w_{\WW} \in \Nul(w_{\WS})$, $w_{\WS} \in \Ndl(w_{\WE})$, and
$w_{\WE} \in \Ndr(w_{\WN})$; we say that $(G,Q)$ is \emph{quasi-triangulated} (see Fig.~\ref{fig:quasi-triangulated}).
 We will then exploit this lemma to prove the main result of the section in Theorem~\ref{th:internally-triangulated-characterization}.
 
\newcommand{\lemmafourconstraineddrawing}{Let $(G,Q)$ be a
  quasi-triangulated \qc graph with $n$ vertices whose corresponding labeling $A_Q$ is angular. If each internal angle of $(G,A_Q)$ has category $0^\circ$ or $90^\circ$, then $(G,Q)$ has a straight-line windrose-planar drawing on the $n \times n$ grid, which can be constructed in $O(n)$ time.}

\begin{lemma}\label{le:4-constrained-drawing}
\lemmafourconstraineddrawing
\end{lemma}

\begin{proof}
We construct a straight-line drawing $\Gamma$ of $(G,Q)$ by assigning the 
$x$-coordinates of the vertices according to a topological ordering
of $G^\rightarrow$ and the $y$-coordinates according to a topological ordering
of $G^\uparrow$. Since, by Lemma~\ref{le:angular-no-source}, $G$ is bi-acyclic, 
such topological orderings can always be found, and the vertices of~$G$ lie 
in $\Gamma$ on the $n\times n$ grid. Clearly, drawing $\Gamma$ can be 
constructed in $O(n)$ time.

We now show that $\Gamma$ is windrose-planar. Let us consider an internal face $f=(u,v,w)$. Since no internal angle has category $180^\circ$ and since $A_Q$ satisfies the Cycle Condition, we have that one of the internal angles of $f$ has category $0^\circ$, say $A(\angle{u,f})=0^\circ$, while the others have $90^\circ$, that is, $A(\angle{v,f})=A(\angle{w,f})=90^\circ$.

We will show that $f$ is drawn in $\Gamma$ with the same orientation as in the planar embedding of $G$, that is, the order of the edges around each of the vertices of $f$ in $\Gamma$ coincides with the one in the given planar embedding of $G$. This is equivalent to saying that the order in which $u$, $v$, and $w$ are encountered when  traversing $f$ clockwise in $\Gamma$ is the same as they appear along the boundary of $f$ in the planar embedding of $G$.
In fact, the orientation of $f$ is uniquely determined by the order of the $x$- and $y$-coordinates imposed by $G^\rightarrow$ and $G^\uparrow$, respectively; see Fig.~\ref{fig:round-triangle}. 
This is due to the fact that, since $A(\angle{u,f})=0^\circ$, vertex $u$ is either a source or a sink of $f$ in both $G^\rightarrow$ and $G^\uparrow$, while, since $A(\angle{v,f})=A(\angle{w,f})=90^\circ$, each of vertices $v$ and $w$ is either a sink or a source in exactly one of $G^\rightarrow$ and $G^\uparrow$, and 
neither a sink nor a source in the other one. 
Hence, the $x$-coordinate of $w$ lies between those of $u$ and $v$, and the $y$-coordinate of~$v$ lies between those of $u$ and $w$, as in Fig.~\ref{fig:round-triangle-c}, or vice versa.
Note that this property would not hold if two interior angles of $f$ had category $0^\circ$ and the third one $180^\circ$, as in Fig.~\ref{fig:round-triangle-a}--\ref{fig:round-triangle-b}, but this case is excluded by hypothesis.

We claim that the orientation of $f$ in $\Gamma$ agrees with the orientation of $f$ in the planar embedding of~$G$. Suppose that $u,w,v$ appear in this clockwise order along $f$ in this embedding. Since $A(\angle{u,f})=0^\circ$, vertices $v$ and $w$ belong to the same quadrant of $u$, say $v,w \in \Nur(u)$. Since $u,w,v$ appear in this clockwise order along $f$, vertex $v$ immediately precedes $u$ in the clockwise order around $w$ in the planar embedding of~$G$. This, together with the fact that $A(\angle{w,f})=90^\circ$, implies that $v \in \Ndr(w)$ (note that $u \in \Ndl(w)$ since $w \in \Nur(u)$). These q-constraints uniquely determine the orientation of the edges of $f$ in $G^\uparrow$ and $G^\rightarrow$, and hence the ordering of the $x$- and $y$-coordinates of $u$, $v$, and $w$ in $\Gamma$. Namely, $x(u) < x(w) < x(v)$ and $y(u) < y(v) < y(w)$. Note that any drawing of $f$ respecting these two orders is such that $u$, $v$, and $w$ appear in this clockwise order (see Fig.~\ref{fig:round-triangle-c}), which concludes the proof of the claim.

\begin{figure}[tb]
   \centering
   \subfigure[]{
   \includegraphics[page=3]{img/round-triangle}
   \label{fig:round-triangle-c}
   }
   \hfil
   \subfigure[]{
   \includegraphics[page=1]{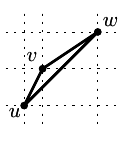}
   \label{fig:round-triangle-a}
   }
   \hfil
   \subfigure[]{
   \includegraphics[page=2]{img/round-triangle}
   \label{fig:round-triangle-b}
   }
   \caption{Orientations of a face $f=(u,v,w)$, when~(a) one angle has category $0^\circ$ and (b--c) two angles have category~$0^\circ$.}
   \label{fig:round-triangle}
\end{figure}

As for the outer face, the fact that its orientation in $\Gamma$ agrees with the one in the planar embedding of $G$ directly follows from the fact that, by the definition of quasi-triangulated \qc graphs, the outer face has a unique orientation, namely $w_{\WW},w_{\WN},w_\WE,w_\WS$ always appear in this clockwise order around the outer face. This agrees with the order in which they appear in $\Gamma$, since $w_\WW$ and $w_\WE$ ($w_\WS$ and~$w_\WN$) are the first and the last vertices in any topological ordering of $G^\rightarrow$ (of $G^\uparrow$).

To complete the proof, we have to argue that having all the faces of $G$ drawn in $\Gamma$ with the same orientation as in the planar embedding of $G$ is sufficient for $\Gamma$ to be planar.
 We define a function $\varphi\colon \mathbb{R}^2\to \mathbb{N}$, where $\varphi(x)$ counts the number of triangles of $G$ bounding an interior face in which a point $x$ is contained; since~$G$ is triangulated, every interior face is bounded by a triangle. 
Since the triangles are oriented consistently, $\varphi(x)$ does not change when $x$ crosses an interior edge: it leaves one triangle and enters another. The function
$\varphi(x)$ changes only (by $\pm 1$) when crossing an external edge,
and we have $\varphi(x)=0$ at infinity. Thus, $\varphi(x)=1$ inside the external quadrilateral (except on edges and vertices) and  $\varphi(x)=0$ outside. We conclude that the triangular faces form a tiling of the external quadrilateral, and thus they form a straight-line planar drawing.
\end{proof}

We now prove that every triangulated \qc graph whose corresponding labeling is angular
can be transformed into a \qc graph that satisfies the conditions of Lemma~\ref{le:4-constrained-drawing}.

\newcommand{\lemmafourconstrainedaugmentation}{Let $(G,Q)$ be a triangulated \qc graph whose corresponding labeling $A_Q$ is angular. Then, there exists a quasi-triangulated \qc graph $(G^*,Q^*)$ such that 
\begin{inparaenum}[(i)]
\item the labeling~$A_{Q^*}$ of $(G^*,Q^*)$ is angular and no internal angle has category $180^\circ$ in $A_{Q^*}$;
\item $G^*$ has at most $3n-1$ vertices; and
\item $G^*$ contains a $1$-subdivision of $G$ as a spanning subgraph. Also, $(G^*,Q^*)$ can be constructed in linear time.
\end{inparaenum} }

\begin{lemma}\label{le:4-constrained-augmentation}
\lemmafourconstrainedaugmentation
\end{lemma}

\begin{proof}
We first need to introduce some additional definitions; refer to Fig.~\ref{fig:lemma2-a}. Let $v$ be any vertex of~$G$. Consider the unique path $P_\uparrow(v) = u_1, \dots, u_m$ such that 
(i) $u_1 = v$, 
(ii) $u_m$ is the unique vertex of $P_\uparrow(v)$ incident to the outer face of $G$, 
and (iii) for each $i = 1, \dots, m-1$, vertex $u_{i+1}$ is the
 rightmost neighbor of $u_i$ in $\Nul(u_i)$ if $\Nul(u_i) \neq \emptyset$;
 otherwise, it is the leftmost neighbor of $u_i$ in $\Nur(u_i)$. 
Since $A_Q$ is angular, no internal angle has category larger than $180^\circ$, and hence  
 $\Nul(u_i) \cup \Nur(u_i) \neq \emptyset$ as long as $u_i$ is not an
external vertex, and since, by Lemma~\ref{le:angular-no-source}, $G^\uparrow$ is acyclic, the path $P_\uparrow(v)$ can be constructed.
Analogously, consider the unique path $P_\rightarrow(v) = w_1,
\dots, w_h$ such that 
(i) $w_1 = v$, (ii) $w_h$ is the unique vertex of $P_\rightarrow(v)$ incident to the outer face of $G$, 
and (iii) for each $i = 1,\dots, h-1$, vertex $w_{i+1}$ is the leftmost neighbor of $w_i$ in $\Ndr(w_i)$ if $\Ndr(w_i) \neq \emptyset$; otherwise, it is the rightmost neighbor of $w_i$ in $\Nur(w_i)$.
Since $\Nur(w_i) \cup\Ndr(w_i) \neq \emptyset$ as long as $w_i$ is not an external vertex and since, by Lemma~\ref{le:angular-no-source}, $G^\rightarrow$ is acyclic, the path $P_\rightarrow(v)$ can be constructed.
Finally, consider the path $\hat P(v)$ from $u_m$ to~$w_h$ obtained by
following the outer face of $G$ clockwise. We denote by $G^\ur(v)$ the subgraph of $G$ induced by the vertices on the boundary and inside the cycle composed of paths $P_\uparrow(v)$, $\hat P(v)$, and $P_\rightarrow(v)$; see Fig.~\ref{fig:lemma2-a}. Graphs~$G^\ul(v)$, $G^\dl(v)$, and $G^\dr(v)$ are defined analogously.

\begin{figure}[tb]
    \centering
    \subfigure[]{
    \includegraphics{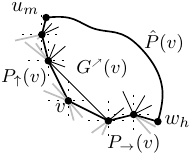}\label{fig:lemma2-a}
    }
    \hfil
    \subfigure[]{
    \includegraphics{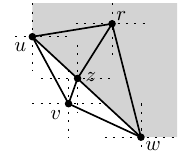}\label{fig:lemma2-b}
    }
    \caption{(a) Graph $G^\capur(v)$ for a vertex $v$ with $\capNur(v) = \emptyset$. Black edges belong to $G^\capur(v)$, while gray edges  do not. Bold edges compose the paths $P_\uparrow(v)$, $\hat P(v)$, and $P_\rightarrow(v)$ delimiting  $G^\capur(v)$. (b) Illustration for the case in which edge $(u,w)$ is not incident to the outer face of $G$. The gray-shaded region is the one where $r$ can lie in  $G$, since $r \notin \capNul(u)$ and $r \notin \capNdr(w)$.}\label{fig:lemma2}  
\end{figure}

We now proceed with the proof of the statement.
Assume that there exists at least one internal angle~$\angle{v,f}$ with category $180^\circ$ in $G$.

We describe the case in which $\Nur(v)=\emptyset$; the other cases are analogous.
Consider the set $X^\ur$ of internal vertices such that for each $x \in X^\ur$ there exists a face $f_x$ with $A(\angle{x,f_x})=180^\circ$ and $\Nur(x)=\emptyset$. Let $v \in X^\ur$ be a vertex such that the rightmost neighbor $u$ of $v$ in $\Nul(v)$ and the leftmost neighbor $w$ of~$v$ in~$\Ndr(v)$ are not in~$X^\ur$. 
Since $G^\ur(u)$ and $G^\ur(w)$ are subgraphs of $G^\ur(v) \setminus v$, it follows that such a vertex exists.

Observe that, since $\Nur(v) = \emptyset$ and since $G$ is triangulated, the triangle $uvw$ is a face of $G$. Also, $u \in \Nul(w)$, by the bi-acyclicity of $(G,Q)$.
  
If edge $(u,w)$ is not incident to the outer face, then there exists a vertex~$r$ creating a face $(u,w,r)$ in~$G$; refer to Fig.~\ref{fig:lemma2-b}. 
We claim that $r \notin \Nul(u)$ and $r \notin \Ndr(w)$.
Namely, $r \in \Nul(u)$ would imply that $\Nur(u) = \emptyset$, since~$w \in \Ndr(u)$ and edges $(u,r)$ and $(u,w)$ are clockwise consecutive around $u$. 
Analogous considerations lead to conclude that $r \notin \Ndr(w)$.

Now, we subdivide $(u,w)$ with a vertex $z$ and augment the obtained $1$-subdivision of $G$ with two dummy edges $(v,z)$ and $(z,r)$ such that $u \in \Nul(z)$, $r \in \Nur(z)$, $w \in \Ndr(z)$, and $v \in \Ndl(z)$.

The obtained \qc graph $(G',Q')$ is internally-triangulated. We claim that the labeling $A_{Q'}$ determined by $Q'$ is angular. By construction, $A_{Q'}$ refines $A_{Q}$, and hence the Vertex Condition is satisfied by all the vertices different from $z$; the Vertex Condition for $z$ holds since all its four incident angles have category $90^\circ$.  Also, all the faces of $G'$ that are not incident to $z$ are also faces of $G$ and hence satisfy the Cycle Condition. Faces $(u,z,v)$ and $(v,z,w)$ have an angle with category $0^\circ$ at $u$ and $w$, respectively, while all the other angles have category $90^\circ$, and hence they satisfy the Cycle Condition; finally, faces $(u,r,z)$ and $(z,r,w)$ always have an angle with category $0^\circ$ and two angles with category $90^\circ$, regardless of the category of the angle at $r$ in $A_Q$.
By Lemma~\ref{le:angular-no-source}, we have that $(G',Q')$ is bi-acyclic and has no internal sources or sinks. 

We remark that the addition of the path between $v$ and $r$ via $z$ adds a constraint on the relative position of $v$ and $r$ in $G'$ that was not in $G$, namely that $r$ has to lie above and to the right of $v$. However, we proved (Lemma~\ref{le:angular-no-source}) that this does not alter the bi-acyclicity of $(G',Q')$. This could also be seen by observing that the presence of edge $(u,w)$ and of the quadrant constraints on $r$ imposed by vertices~$u$ and~$w$ implies that there exists no path in $G$ from $r$ to $v$ whose edges are either all directed upwards or all directed rightwards.

Finally, note that the two angles incident to $v$ that are delimited by edge $(v,z)$ get category $90^\circ$, since $z \in \Nur(v)$, and hence $A_{Q'}$ has fewer internal angles with category $180^\circ$ than $A_Q$.

The case in which the edge $(u,w)$ is incident to the outer face is simpler.
We subdivide $(u,w)$ with a vertex $z$, augment the
obtained $1$-subdivision of $G$ with a dummy edge $(v,z)$, and set $z \in \Nur(v)$, $z \in \Nul(w)$, and $z \in \Ndr(u)$; hence
obtaining a new internally-triangulated \qc graph $(G',Q')$. The proof that~$A_{Q'}$ is angular and has fewer internal angles with category $180^\circ$ than $A_Q$ proceeds as above.

Since in both cases we reduced the number of angles with category $180^\circ$ by at least one, iterating the above construction yields a \qc graph $(G^+,Q^+)$ whose corresponding labeling $A_{Q^+}$ is angular and has no angle with category $180^\circ$. Note that the number of angles with category $180^\circ$ is bounded by the number of internal faces of $G$, which is at most $2n-5$; in fact, each internal face can contain at most one of such angles, due to the Cycle condition. This implies that the algorithm performs at most $2n-5$ iterations, in order to obtain $G^+$ starting from $G$.

We claim that $G^+$ contains a $1$-subdivision of $G$ as a spanning subgraph. To prove this, observe that new vertices are added to $G$ only as subdivision vertices; also, once a subdivision edge or a dummy edge has been added at some step, it is not subdivided in any of the following steps. For the subdivision edge $(w,z)$, this depends on the fact that the (at most) two vertices~$v$ and $r$ creating a face with $w$ and $z$ are such that $z \in \Nur(v)$, $w \in \Ndr(v)$, $z \in \Ndl(r)$, and either $w \in \Ndl(r)$ or $w \in \Ndr(r)$; the proof for the other edges is analogous. Since the algorithm performs at most $2n-5$ iterations, and since each iteration introduces at most one subdivision vertex, we have that the number of such vertices is at most $2n-5$. Thus, $G^+$ contains at most $3n-5$ vertices. 

Further, $G^+$ is internally-triangulated and its outer face is a $1$-subdivision of the outer face of $G$ such that, if $d$ is the subdivision vertex of edge $(a,b)$ and $b \in \Npar(a,\circ)$, for some $\circ \in \sets$, then $d \in \Npar(a,\circ)$ and~$b \in\Npar(d,\circ)$. 

We augment $(G^+,Q^+)$ to a quasi-triangulated \qc
graph $(G^*,Q^*)$ by adding four vertices $w_\WN$, $w_\WW$, $w_\WS$, and
$w_\WE$ in such a way that they respect the required conditions of the definition of a quasi-triangulated \qc graph. 
Also, we add edges between the vertices on the outer face
 of~$(G^+,Q^+)$ and $w_\WN$, $w_\WW$, $w_\WS$, $w_\WE$ in such a way that each vertex of the outer face of~$(G^+,Q^+)$ has at least a neighbor in each of its quadrants, as follows (refer to Fig.~\ref{fig:quasi-triangulated}). For each vertex $v$ on the outer face of~$(G^+,Q^+)$ such that $\Nur(v)=\emptyset$ ($\Nul(v)=\emptyset$, $\Ndl(v)=\emptyset$, or $\Ndr(v)=\emptyset$), we connect $v$ to $w_{E}$ (to $w_N$, to $w_W$, or to $w_S$, respectively), and add $w_E$ to $\Nur(v)$ ($w_N$ to $\Nul(v)$, $w_W$ to $\Ndl(v)$, or $w_S$ to $\Ndr(v)$, respectively). 

Note that $G^*$ has at most $3n-1$ vertices, since it has been obtained from $G^+$ by adding four vertices. Also, by construction, all the internal angles at the vertices on the outer face of~$(G^+,Q^+)$ get labels at most $90^\circ$ in the labeling $A_{Q^*}$ 
determined by $Q^*$.  
Since the four external vertices satisfy the Vertex Condition, and since $A_{Q^*}$ refines $A_{Q}$, we have that the Vertex Condition is satisfied by all the internal vertices.
Further, all the faces of $G^*$ that are not incident to the four external vertices are also faces of $G^+$ and hence satisfy the Cycle Condition. 
Faces of $G^*$ that are incident to exactly one of $w_\WN$, $w_\WW$, $w_\WS$, $w_\WE$, say $w_\WN$, have the angle at $w_\WN$ with category $0^\circ$ and the two other angles with category $90^\circ$. Finally, the faces that are incident to two clockwise-consecutive external vertices, say $w_\WN$ and $w_\WE$, have the angle at $w_\WN$ with category $0^\circ$ and the two other angles with category $90^\circ$. Hence, all the faces satisfy the Cycle Condition.
By Lemma~\ref{le:angular-no-source}, we have that $(G^*,Q^*)$ is bi-acyclic and has no internal sources or sinks. 

We now show that the augmentation of $(G,Q)$ to $(G^*,Q^*$) can be performed in linear time. Clearly, we can construct $P_\uparrow(v)$, $P_\leftarrow(v)$, $P_\downarrow(v)$, and $P_\rightarrow(v)$ for all vertices in total linear time via a depth-first search. In fact, observe that a path $P_{\scriptstyle\triangle}(v)$ of a vertex $v$, with ${\scriptstyle\triangle} \in \{\uparrow, \leftarrow, \downarrow, \rightarrow\}$, passing through a vertex~$u$ entirely contains path $P_{\scriptstyle\triangle}(u)$. 
For each $\circ \in \sets$, consider the set $X^\circ$ of internal vertices such that for each $x \in X^\circ$ there exists a face $f_x$ with $A(\angle{x,f_x})=180^\circ$ and $\Npar(x,\circ)=\emptyset$. Let $\circ=\ur$; the other cases are analogous. Until $X^\ur = \emptyset$, we select a vertex $v \in X^\ur$ and test whether the rightmost neighbor $u$ of $v$ in $\Nul(v)$ and the leftmost neighbor $w$ of~$v$ in~$\Ndr(v)$ are not in $X^\ur$. If this is the case, we perform the augmentation described above, which can be easily performed in constant time, and remove~$v$ from~$X^\ur$. Otherwise, at least one of $u$ and $w$ belongs to $X^\ur$, say $u$, and we recursively repeat the test starting from $u$. It follows from the fact that $G^\ur(u)$ and $G^\ur(w)$ are subgraphs of $G^\ur(v) \setminus v$, that each vertex in $X^\ur$ appears exactly once in the recursion tree.
This concludes the proof of the lemma.
\end{proof}

We now present the main results of the section.

\newcommand{\thinternallytriangulatedcharacterization}{A triangulated
  \qc graph $(G,Q)$ is windrose-planar if and only if the
  labeling~$A_Q$ determined by its q-constraints is angular. Also, if
  $(G,Q)$ is windrose-planar and has $n$ vertices, then it admits a $1$-bend windrose-planar drawing on the $3n \times 3n$ grid, which can be constructed in $O(n)$ time and has at most $2n-5$ bends.}

\begin{theorem}\label{th:internally-triangulated-characterization}
  \thinternallytriangulatedcharacterization
\end{theorem}

\begin{proof}
The necessity comes from Observation~\ref{obs:windrose-to-labeling} and Lemma~\ref{lem:angular-labeling-necessary}.
We prove the sufficiency. First, apply Lemma~\ref{le:4-constrained-augmentation} to construct in linear time a quasi-triangulated \qc graph $(G^*,Q^*)$ such that
\begin{inparaenum}[(i)]
  \item the labeling~$A_{Q^*}$ of $(G^*,Q^*)$ is angular and no internal angle has category $180^\circ$ in $A_{Q^*}$;
  \item $G^*$ has at most $3n-1$ vertices; and
  \item $G^*$ contains a $1$-subdivision of $G$ as a spanning subgraph.
\end{inparaenum}

The first two conditions imply that we can apply Lemma~\ref{le:4-constrained-drawing} to construct in linear time a straight-line windrose-planar drawing $\Gamma^*$ of $(G^*,Q^*)$ on the $3n \times 3n$ grid. 
We construct a $1$-bend windrose-planar drawing $\Gamma$ of
$(G,Q)$ starting from $\Gamma^*$, as follows. Initialize $\Gamma$ as $\Gamma^*$ restricted to the vertices and edges of $G$. Then, for each edge $(u,w)$ of $G$ that is not in $G^*$, consider the subdivision vertex $z$ of $(u,w)$ in $G^*$. Draw edge $(u,v)$ in $\Gamma$ with a $1$-bend poly-line whose two straight-line segments coincide with the drawing of the edges $(u,z)$ and $(z,w)$ in $\Gamma^*$. Edge $(u,w)$ is crossing-free in~$\Gamma$ since~$(u,z)$ and $(z,w)$ are crossing-free in~$\Gamma^*$. Also, it is drawn in~$\Gamma$ as an $xy$-monotone curve and~$u$ lies in the correct quadrant of~$w$ due to the fact that~$\Gamma^*$ is a straight-line drawing and that, by construction, if $u \in \Npar(w,\circ)$ in~$G$, for some $\circ \in \sets$, then $z \in \Npar(w,\circ)$ and $u \in \Npar(z,\circ)$ in~$G^*$. Since each bend of $\Gamma$ corresponds to a subdivision vertex of $G^*$, we have that $\Gamma$ contains at most $2n-5$ bends. Also, all the vertices and the bend-points of all the edges of $G$ lie on the $3n \times 3n$ grid in $\Gamma$, since they correspond to vertices of $G^*$, which lie on such a grid in $\Gamma^*$.
\end{proof}

As observed in Section~\ref{se:upward}, {\sc Windrose Planarity} is a generalization of {\sc Upward Planarity}. Hence, our results also extend to {\sc Upward Planarity}. In particular, Theorem~\ref{th:internally-triangulated-characterization} gives a new proof of the result by Di Battista, Tamassia, and Tollis~\cite{dtt-arsdud-92}, who presented an $O(n)$-time algorithm for constructing 1-bend upward planar drawings with at most $2n - 5$ bends on a grid with $O(n^2)$ area. By exploiting bitonic st-orderings, Gronemann~\cite{grone-bstoupg-16} has recently improved on this result; in fact, he has shown that every upward planar $n$-vertex graph admits a 1-bend upward planar drawing within quadratic area having at most $n - 3$ bends in total. Investigating the power of bitonic st-orderings to construct 1-bend windrose-planar drawings with small number of bends in quadratic area represents an intriguing research direction.

The characterization provided in Theorem~\ref{th:internally-triangulated-characterization} directly yields a linear-time testing algorithm for the class of triangulated \qc graphs, as the conditions of the theorem can be tested efficiently.

\begin{theorem}\label{th:internally-triangulated-test}
It is possible to test in $O(n)$ time whether a triangulated \qc graph
$(G,Q)$ with $n$ vertices is windrose-planar and, if so, to build a 1-bend windrose-planar drawing of it with at most $2n-5$ bends on the $3n \times 3n$ grid.
\end{theorem}

\section{Testing windrose planarity with fixed planar embedding}\label{sec:necessary}

We now extend the results of Section~\ref{se:triangulated} to general \qc plane graphs. Namely, we show that a \qc plane graph is windrose-planar if and only if its q-constraints determine an angular labeling. However, while a triangulated \qc graph admits a unique labeling, a general \qc plane graph may determine several labelings, one for each large-angle assignment, as discussed in Section~\ref{sec:angle-categories}.

We first show that a large-angle assignment for all vertices whose corresponding labeling is angular, if one exists, can be found via a simple flow network that is inspired by the one devised 
by Bertolazzi,
Di Battista, Liotta, and Mannino~\cite{bdlm-udtg-94} 
for testing upward-planarity. The proof of the next lemma is based on \mbox{such a flow network}.

\begin{lemma}
  \label{lem:large-angle-assignment}
Given a \qc plane graph $(G,Q)$, it can be determined in $O(n \log^3 n)$ time whether there exists a large-angle assignment $L$ such that the corresponding labeling $A_{Q,L}$ is angular.
\end{lemma}

\begin{proof}
Recall that the angle assignment around a vertex is unique, except for the ambiguous vertices of $(G,Q)$.
For such vertices, by the Vertex Condition, exactly one of their incident angles needs to be assigned an angle category of $360^\circ$.
Consider an internal face $f$ whose facial cycle has length $k$.  By
the Cycle Condition the angles incident to $f$ must sum up to $k \cdot
180^\circ - 360^\circ$.  Since we know the angles at all non-ambiguous
vertices, we can compute a demand $d(f)$ of how many ambiguous vertices
must assign an angle of $360^\circ$ to $f$ such that the Cycle
condition is satisfied for $f$.  
Similarly, a demand can be computed for the outer face.   
Clearly, it is a necessary condition that $d(f)$ is a non-negative integer.

Altogether, we thus need to find an assignment of large angles
($360^\circ$) of ambiguous vertices to faces such that each ambiguous
vertex assigns a large angle to exactly one incident face and such that each
face~$f$ receives $d(f)$ large angles.  We model this as a flow
network.  Let $F$ be the set of faces of $G$ and let~$B$ denote the
set of ambiguous vertices.  The flow network $N$ has vertex set $F
\cup B$ and it contains arcs with capacity 1 connecting each vertex $b
\in B$ to its incident faces.  The vertices in $B$ are sources with
maximum out-flow~1, the vertices in $F$ are sinks with maximum in-flow
$d(f)$.
By construction, the maximum flows of $N$ where every vertex in $B$
has out-flow~1 and each face $f$ has in-flow $d(f)$ correspond
bijectively to the large-angle assignments of $(G,Q)$ that result in an
angular labeling.

Observe that $N$ is planar, and hence we can use the algorithm by
Borradaile,
Klein, Mozes, Nussbaum, and Wulff-Nilsen~\cite{bkmn-mmmdp-11} to compute such a flow 
in $O(n\log^3 n)$ time.  To control the maximum in-flow and maximum out-flow
of vertices, we simply connect a new source vertex to each source and we
connect each sink to a new sink vertex, so that  
the capacities of their incident arcs can be used to limit the maximum in- and out-flows.
\end{proof}

In the following, let $G$ be a plane graph with an angular labeling
$A$.  We show how to augment $G$ to a triangulated plane graph $G'$
with an angular labeling $A'$ that refines $A$.  Then, an angular
drawing of~$(G,A)$ can be obtained from an angular drawing of $(G',A')$, which
exists by Theorem~\ref{th:internally-triangulated-characterization}.

\newcommand{\lemmatriangulateangularlabeling}{Let $G$ be a plane graph with an angular labeling $A$.  Then, $G$ can
  be augmented in linear time to a triangulated plane graph $G'$ with angular labeling $A'$ that refines $A$.}

\begin{lemma}
  \label{lem:triangulate-angular-labeling}
  \lemmatriangulateangularlabeling
\end{lemma}

\begin{proof}
  If the outer face of $G$ is not a triangle, let $e$ and $e'$ be two consecutive edges on the outer face of $G$, sharing a vertex $v$.
  We add a new triangle $(a,b,c)$ that contains $G$ in its interior and an edge $(v,a)$, and we set 
  $A(\angle{(a,b),(a,v)})=A(\angle{(a,v),(a,c)})=0$, 
  $A(\angle{(b,c),(b,a)})=A(\angle{(c,a),(c,b)})=90^\circ$, 
  $A(\angle{e,(v,a)})=0$, and $A(\angle{(v,a),e'})=A(\angle{e,e'})$. Clearly, labeling $A$ is still angular.

The original outer face has been turned into an additional interior face of this new graph, again denoted by $G$. Note that this graph is connected and its outer face is a triangle. 

In the following, we will iteratively add edges until $G$ becomes triangulated. We show that, as long as $G$ is not triangulated, we can add an edge $e$ to $G$ such that the resulting graph $G+e$ is plane and admits an angular labeling $A'$ that refines $A$.  The lemma follows by induction.

Let $f$ be an internal face of $G$ whose facial cycle has length $k\ge 4$. We have the following claim.

\begin{claimx}
  Let $\alpha_1,\alpha_2,\ldots,\alpha_k$ be a cyclic sequence of
$k\ge4$
  angles from the range
$\{0^\circ, 90^\circ, 180^\circ, 270^\circ,$ $360^\circ\}$
with sum $k \cdot 180^\circ-360^\circ$. Such a sequence is either the sequence
$(90^\circ, 90^\circ,90^\circ, 90^\circ)$ of length $4$, or it
 must contain an adjacent
 pair $(\gamma,\delta)$ with 
    $\gamma\in\{180^\circ,270^\circ,360^\circ\}$ and 
    $\delta\in\{0^\circ,90^\circ\}$ \mbox{(or vice versa).}
\end{claimx}
\begin{proof}
  Observe that the average angle is less than $180^\circ$.  Hence, if
  one of the angles above average, i.e., $180^\circ, 270^\circ, 360^\circ$,
  occurs at all, some of the other two angles $0^\circ, 90^\circ$ must
  also occur, and somewhere the two classes of angles must appear in
  adjacent positions.  We are left
  with the case that only angles $0^\circ$ or $90^\circ$ appear; this
  allows only the cyclic sequences
  $(90^\circ+90^\circ+90^\circ+90^\circ)$, $(90^\circ+90^\circ+0^\circ)$, and
  $(0^\circ+0^\circ)$, but the last two are excluded because they have
  \mbox{less than four elements.}
\end{proof}
\begin{figure}[tb]
  \centering
  \subfigure[]{
  \includegraphics[page=1,height=0.15\textwidth]{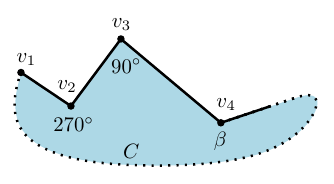}
  \label{fig:ear-cut-1}
  }\hfil
  \subfigure[]{
  \includegraphics[page=2,height=0.15\textwidth]{img/ear-cut}
  \label{fig:ear-cut-2}
  }
  \caption{
(a) A face $f$ with four consecutive vertices $v_1v_2v_3v_4$.
(b) Cutting a triangular ear $v_2v_3v_4$ from $f$. The dashed edges are the new edge $v_2v_4$ and parallel edges that might already exist in the graph.}
  \label{fig:ear-cut}
\end{figure}

Then, one of the patterns of the previous claim must occur, and, in each case, we will show that we can cut off a triangular ``ear'' from $f$.

We will illustrate the operation for the pattern $(270^\circ, 90^\circ)$. Together with the adjacent angles, we have a subsequence $(\alpha, 270^\circ, 90^\circ,\beta)$ at four vertices $v_1,v_2,v_3,v_4$. We will add a \emph{diagonal} edge between $v_2$ and $v_4$, transforming the angle sequence into $(\alpha, 180^\circ+90^\circ, 90^\circ,0^\circ+\beta)$. This sequence indicates how the angles at $v_2$ and $v_4$ are split by the new edge; see Fig.~\ref{fig:ear-cut}.
 The angles in the triangle $v_2v_3v_4$ are $(90^\circ, 90^\circ,0^\circ)$ (the subsequence between the two $+$ signs). According to Lemma~\ref{lem:modifications}, it suffices to check that the angles were correctly split:  $270^\circ=180^\circ+90^\circ$ and $\beta=0^\circ+\beta$; and that the triangular face has the correct angle sum: $90^\circ+ 90^\circ+0^\circ =180^\circ$. This means that we have split the $k$-cycle into a triangle and a $(k-1)$-cycle such that
the Vertex and Cycle Conditions are satisfied (Lemma~\ref{lem:modifications}).
It remains to show that $(v_2,v_4)$ did not create a multiple edge.
 Fig.~\ref{fig:ear-cut-2} shows the two possible ways how $(v_2,v_4)$
 might already be connected by an edge running outside the face $f$.
In both cases, consider the internal face $g$ of the resulting $2$-gon $(v_2,v_4)$ after removing its interior. By Corollary~\ref{cor:subgraph-labeling}, it satisfies the Cycle Condition, that is, all its internal angles have category $0^\circ$. On the other hand, it has a nonzero angle at~$v_2$, namely, at least $90^\circ$ if $v_3$ is inside $g$, and at least $180^\circ$, otherwise. This is a contradiction. 
 It is also obvious that $v_2$ and $v_4$ are distinct vertices, because otherwise the angle at $v_3$ would have to be $360^\circ$.

We represent the above example and three other cases in tabular form
in Table~\ref{tab:ear}. The treatment is similar in all cases. The
crucial property for excluding parallel edges is that both new angles
at $v_2$ are always positive.

Table~\ref{tab:ear2} shows the three remaining cases of the previous
claim. Here, we use a different argument to avoid
multiple edges: we always have two choices for inserting a diagonal,
$(v_1,v_3)$ or $(v_2,v_4)$. By planarity, these pairs cannot \emph{both} be
connected outside $C$.

\begin{table*}[tb]
  \centering
  \caption{Cutting off an ear from a face.}
\subtable[The four cases involving angles $270^\circ$ and $360^\circ$. \label{tab:ear}]{
\hspace{5mm}
  \begin{tabular}[t]{|cccc|}
\hline
    &$v_2$& $v_3$ & $v_4$\\
\hline
\hline
    &$270^\circ$,& $90^\circ$, & $\beta$\\
\hline 
$\Rightarrow$&
$180^\circ+90^\circ$,&
$90^\circ$,
&$0^\circ+\beta$\\

\hline
\hline
    &$360^\circ$,& $90^\circ$, & $\beta$\\
    
\hline
$\Rightarrow$&
$270^\circ+90^\circ$,&
$90^\circ$,
&$0^\circ+\beta$\\
\hline
\hline
&$270^\circ$,& $0^\circ$, & $\beta$\\
\hline
$\Rightarrow$&
$90^\circ+180^\circ$,&
$0^\circ$,
&$0^\circ+\beta$\\
\hline
\hline
&$360^\circ$,& $0^\circ$, & $\beta$\\
\hline
$\Rightarrow$&
$180^\circ+180^\circ$,&
$0^\circ$,
&$0^\circ+\beta$\\
\hline
  \end{tabular}
  \hspace{5mm}
}
\hfil
\subtable[The three cases not involving
angles $270^\circ$ and $360^\circ$. \label{tab:ear2}]{
\begin{tabular}[t]{|ccccc|}
\hline
    &$v_1$    &$v_2$& $v_3$ & $v_4$\\
\hline
\hline
    &$\alpha$,& $180^\circ$,& $0^\circ$, & $\beta$\\
\hline
$\Rightarrow$
&$\alpha+0^\circ$,& $180^\circ$,& $0^\circ+0^\circ$, & $\beta$\\
$\Rightarrow$
    &$\alpha$,& $0^\circ+180^\circ$,& $0^\circ$, & $0^\circ+\beta$\\
\hline
\hline
    &$\alpha$,& $180^\circ$,& $90^\circ$, & $\beta$\\
\hline
$\Rightarrow$
&$\alpha+0^\circ$,& $180^\circ$,& $0^\circ+90^\circ$, & $\beta$\\
$\Rightarrow$
    &$\alpha$,& $90^\circ+90^\circ$,& $90^\circ$, & $0^\circ+\beta$\\
\hline
\hline
    &$90^\circ$,& $90^\circ$,& $90^\circ$, & $90^\circ$\\
\hline
$\Rightarrow$
    &$90^\circ+0^\circ$,& $90^\circ$,& $90^\circ+0^\circ$, & $90^\circ$\\
$\Rightarrow$
    &$90^\circ$,& $0^\circ+90^\circ$,& $90^\circ$, & $0^\circ+90^\circ$\\
\hline
  \end{tabular}
} 
\end{table*}

To show that the augmentation of $G$ can be performed in linear time, we describe how to augment a face $f$ in time linear in the size of $f$. We consider the angles around $f$ in clockwise order and push them onto a stack. After each push operation, we check whether the three top-most angles or the four top-most angles on the stack match one of the patterns in Table~\ref{tab:ear} or~\ref{tab:ear2}, respectively; if this is the case, then we apply the augmentation and the corresponding transformation of the stack. This transformation updates the three top-most or the four top-most positions of the stack and reduces the stack size by one. We exhaustively perform such reductions before we push the next angle onto the stack. Clearly, the number of reductions is bounded by the number of push operations and each reduction takes constant time.
\end{proof}

By using Lemma~\ref{lem:triangulate-angular-labeling}, we can prove the following.

\newcommand{\lemmaplanecharacterization}{A \qc plane graph $(G,Q)$ is windrose-planar if and only if it admits a large-angle assignment $L$ whose corresponding labeling $A_{Q,L}$ is angular.
Also, if $(G,Q)$ is windrose-planar and $A_{Q,L}$ is given, then it admits a $1$-bend windrose-planar drawing with at most $2n-5$ bends on the $3n \times 3n$ grid, which can be computed in $O(n)$ time, where $n$ is the number of vertices of $G$.}

\begin{theorem}\label{th:plane-characterization}
\lemmaplanecharacterization
\end{theorem}

\begin{proof}
For the sufficiency, assume that there exists a windrose-planar drawing $\Gamma$ of $(G,Q)$. Then, $\Gamma$ defines a large-angle assignment $L$. By Observation~\ref{obs:windrose-to-labeling}, $Q$ and $L$ uniquely determine a labeling $A_{Q,L}$. Since $\Gamma$ is an angular drawing of $(G,A_{Q,L})$, by Lemma~\ref{lem:angular-labeling-necessary}, $A_{Q,L}$ is angular.

We now prove the necessity. Let $(G,Q)$ be a \qc plane graph with a large-angle assignment $L$ such that $A_{Q,L}$ is angular. Then, by Lemma~\ref{lem:triangulate-angular-labeling}, $G$ can be augmented to a triangulated plane graph $G'$ with angular labeling $A'$ that refines $A_{Q,L}$ in $O(n)$ time. Note that this augmentation only adds edges, so $G'$ has $n$ vertices. 
Let $Q'=Q_{A'}$ be the set of q-constraints determined by $A'$ (Observation~\ref{obs:labeling-to-constraints}). 
By Theorem~\ref{th:internally-triangulated-characterization}, $(G',Q')$ admits a $1$-bend windrose-planar drawing $\Gamma'$ on the $3n \times 3n$ grid, which can be constructed in $O(n)$ time and has at most $2n-5$ bends.
Since all the q-constraints in $Q$ also belong to $Q'$, the drawing $\Gamma$ of $(G,Q)$ obtained by removing edges and vertices in $G' \setminus G$ from $\Gamma'$ is a $1$-bend windrose-planar drawing with the same properties as $\Gamma'$. 
\end{proof}

We are now ready to prove the main result of the section.
By Lemma~\ref{lem:large-angle-assignment}, it is possible to compute in $O(n \log^3 n)$ time a large-angle assignment $L$ for $(G,Q)$, if one exists, such that the corresponding labeling $A_{Q,L}$ is angular. Then, if $L$ exists, we can construct in linear time a drawing of $(G,Q)$ with the properties described in Theorem~\ref{th:plane-characterization}. We summarize this result in the following theorem.

\newcommand{\maintheorem}{
In $O(n \log^3 n)$ time, it is possible to test whether a \qc plane graph is windrose-planar and, if so, to construct a 1-bend windrose-planar drawing of it with at most $2n-5$ bends on the $3n \times 3n$  grid.}

\begin{theorem}\label{th:main-theorem}
\maintheorem
\end{theorem}

\section{Straight-line realizability of windrose-planar graphs}
\label{se:straight-line}

In Theorem~\ref{th:plane-characterization}, we proved that any windrose-planar \qc graph can be realized with one bend per edge. In this section, we ask whether this is possible even with straight-line edges. We remark that every upward-planar directed graph admits a straight-line upward-planar drawing~\cite{efln-sldahgcg-06}.

In Theorem~\ref{th:3-trees} we answer the above question in the positive for a particular class of graphs, in which every block is either an edge or a planar $3$-tree. Note that this class also includes trees as a subclass. 
On the other hand, in Theorem~\ref{th:exparea}, we give a family of
\qc graphs that require exponential area if drawn straight-line, and
in Theorem~\ref{th:bimonotone}, we provide a negative result for the
straight-line realizability in a setting that is strongly related to
the one we study. This answers an open question posed by Fulek, Pelsmajer, Schaefer and
 \v{S}tefankovi\v{c}~\cite{fpss-htmd-11}.

\subsection{Straight-line windrose-planar drawings}

We first present the two results concerning the existence and the area-requirements of straight-line windrose-planar drawings of \qc graphs.

\newcommand{\lemmathreeetrees}{Every windrose-planar \qc graph $(G,Q)$ whose blocks are either edges or planar $3$-trees admits a straight-line windrose-planar drawing.}

\begin{theorem}\label{th:3-trees}
\lemmathreeetrees
\end{theorem}

\begin{proof}
Let $\mathcal{E}$ be the planar embedding of $(G,Q)$ in a windrose-planar drawing of $(G,Q)$.
We show how to compute a straight-line windrose-planar drawing $\Gamma$ of $(G,Q)$ with the same embedding as $\mathcal{E}$. The construction is performed inductively on the number $n$ of vertices of $(G,Q)$. 

We have two base cases, namely when $n=2$ and when $n=3$. In the first case, $(G,Q)$ is an edge, and hence it always admits a straight-line windrose-planar drawing. In the second case, $(G,Q)$ is either a path, and hence can be realized using straight-line edges, or a $3$-cycle $\triangle$, and, since it is windrose-planar, admits an angular labeling. Hence, the fact that $\triangle$ satisfies the Cycle Condition implies that it admits a straight-line windrose-planar drawing. In particular, if two angle categories are $90^\circ$ and the third angle category is $0^\circ$, then we can make $\triangle$ quasi-triangulated and use Lemma~\ref{le:4-constrained-drawing}. Otherwise, there is a $180^\circ$ angle category and two $0^\circ$ angle categories; then, we can draw one edge with slope $\pm 1$ and place the vertex incident to the $180^\circ$ angle category close enough to the edge to make the drawing angular.

In the inductive case, $n>3$. Consider a block $\beta$ of $(G,Q)$ such that, if $(G,Q)$ is biconnected, then~$\beta=(G,Q)$; otherwise, $\beta$ has a single cutvertex $c$ incident to it and not all the edges of $(G,Q)$ that are incident to the outer face of $(G,Q)$ in $\mathcal{E}$ belong to $\beta$. Note that such a block exists. In fact, if $(G,Q)$ is not biconnected, then there exist at least two blocks with a single cutvertex incident to them, and thus at least one of these blocks does not contain all the edges of $(G,Q)$ that are incident to the outer face in $\mathcal{E}$. Finally note that, for a block $\beta$ satisfying the above properties, we have that the cut-vertex $c$ is incident to the outer face of $\beta$ in $\mathcal{E}$ and all the internal faces of $\beta$ are also faces of~$(G,Q)$.

We distinguish three cases, based on whether $|\beta|=2$, $|\beta|=3$, or $|\beta| > 3$.

If $|\beta|=2$, then let $v$ be the vertex of $\beta$ different from
$c$. Note that $v$ has degree $1$ in $(G,Q)$. Clearly, graph
$G'=G \setminus v$ is still such that every block is either an edge or
a planar $3$-tree. We inductively compute a straight-line
windrose-planar drawing $\Gamma'$ of $(G',Q')$ with the same embedding
as $\mathcal{E}$ restricted to the vertices and edges of $(G',Q')$. Let $z$ 
and $w$ (possibly $z=w$) be the neighbors of $c$ that precede and
follow $v$, respectively, in the clockwise order of the
neighbors around $c$ in $\mathcal{E}$; see Fig.~\ref{fig:simplecase-3tree}.

Let $D$ be a disk centered at $c$ of radius sufficiently small not to contain 
any vertex of $(G',Q')$ different from $c$ in its interior. Let $\alpha$ be 
the subregion of $D$ delimited by the two straight-line edges~$(c,w)$ and~$(c,z)$ in 
$\Gamma'$ and by the part of the boundary of $D$ from the intersection with 
edge $(c,z)$ to the one with edge $(c,w)$ in clockwise direction. Note that 
placing $v$ in any interior point of $\alpha$ yields a straight-line 
crossing-free drawing of $G$ whose embedding coincides with $\mathcal{E}$; refer to 
Fig.~\ref{fig:simplecase-3tree}. 
We claim that there exists a point inside $\alpha$ such that placing $v$ on 
this point yields a straight-line windrose-planar drawing of $(G,Q)$.
Let~$\circ$,~$\times$, and $\scriptstyle\triangle$ be the three indexes such 
that $v \in \Npar(c,\circ)$, $w \in \Npar(c,\times)$, and 
$z \in \Npar(c,\scriptstyle\triangle)$. Since $\mathcal{E}$ has been obtained
from a windrose-planar drawing of $G$, it holds, by 
Observation~\ref{obs:planarity}, that either (i) $\circ=\times$, 
(ii) $\circ=\scriptstyle\triangle$, or (iii) $\circ$ is encountered before 
$\scriptstyle\triangle$ in the circular sequence $\sets$ starting 
from $\times$. In all of these cases, the $\circ$-quadrant of $c$ has a non-empty
intersection with $\alpha$, and the claim follows.

If $|\beta|=3$, then let $u$ and $v$ be the vertices of $\beta$ different from $c$, with $v$ preceding $u$ in the clockwise order of the neighbors around $c$ in $\mathcal{E}$. Inductively compute a straight-line windrose-planar drawing $\Gamma'$ of $(G'=G \setminus \{v,u\},Q')$ with the same embedding as $\mathcal{E}$ restricted to the vertices and edges of $G'$. Let $w$ and $z$ (possibly $w=z$) be the neighbors of $c$ that precede $v$ and follow $u$, respectively, in the clockwise order of the neighbors around $c$ in $\mathcal{E}$.
As in the previous case, the fact that $\mathcal{E}$ has been obtained from a windrose-planar drawing of $(G,Q)$ ensures that the region $\alpha$, defined as above, intersects the $\circ$-quadrant and the $\times$-quadrant of $c$, where $u \in \Npar(c,\circ)$ and $v \in \Npar(c,\times)$. For the same reason, either $\circ = \times$ or there exists at most one quadrant between $\circ$ and $\times$ in the circular sequence $\sets$ starting from $\circ$. This implies that there exist pairs of points inside $\alpha$, one in the $\circ$-quadrant and one in the $\times$-quadrant of $c$, such that, when placing $u$ and $v$ on these points, the angle spanned by rotating around $c$ in clockwise direction from edge~$(c,v)$ to edge $(c,u)$ is smaller than $180^\circ$, which implies that edge $(u,v)$ can be drawn as a straight-line crossing-free segment while respecting embedding $\mathcal{E}$. By choosing such an appropriate pair of  points for $u$ and~$v$, one can also place $u$ in the correct quadrant of $v$, and vice versa. In fact, if there exists a quadrant between $\circ$ and $\times$ in the circular sequence $\sets$ starting from $\circ$, then $u$ lies in the $\circ$-quadrant of $v$ for any placement of $u$ and $v$ in the correct quadrants of $c$; since this is true even if the edges are not required to be straight-line, it is also true in the windrose-planar drawing of $(G,Q)$ we used to compute $\mathcal{E}$; hence, we have that $u \in \Npar(v,\circ)$. If either $\circ$ and $\times$ are consecutive in the circular sequence $\sets$ or~$\circ = \times$, then the correct relative position of $u$ and $v$ can be obtained by adjusting the length of the edges~$(c,u)$ and~$(c,v)$, and by the angle they form.

\begin{figure}[tb]
  \centering
  \begin{minipage}[t]{.47\textwidth}
  \centering
  \subfigure[]{
    \includegraphics[page=1]{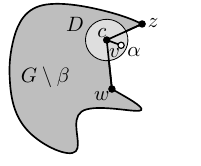}\label{fig:simplecase-3tree}
  }\hfil
  \subfigure[]{
    \includegraphics[page=2]{img/simplecase-3tree}\label{fig:edgeconfiguration-23trees}
  } 
  \caption{Adding a vertex $v$ of a leaf block $\beta$ in a drawing $\Gamma'$ of $(G',Q')$. (a) $\beta=(c,v)$ with $c \in \capNul(v)$ and (b) $\beta$ is a $3$-tree. }
  \end{minipage}
  \hfill
\begin{minipage}[t]{.47\textwidth}
\centering
   \subfigure[]{\label{fig:triangle-3tree-a}
   \includegraphics[page=2]{img/triangle-3tree}
   }\hfil
   \subfigure[]{\label{fig:triangle-3tree-b}
   \includegraphics[page=1]{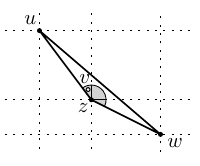}
   }
  \label{fig:triangle-3tree}
  \caption{The two cases for the placement of vertex $v$ in the interior of triangle $uwz$ in the proof of Theorem~\ref{th:3-trees}.}
\end{minipage}
\end{figure}

Otherwise, $|\beta| > 3$ and hence $\beta$ is a planar $3$-tree. Thus, there exists an internal vertex $v$ of $\beta$ of degree $3$. Note that $v \neq c$, since $c$ is incident to the outer face of $\beta$. Hence, $v$ has degree $3$ in $G$ as well, and graph $G'=G \setminus v$ is still a graph of which each block is either an edge or a planar 3-tree. Inductively compute a straight-line windrose-planar drawing $\Gamma'$ of $(G',Q')$ with the same embedding as $\mathcal{E}$ restricted to the vertices of $G'$; refer to Fig.~\ref{fig:edgeconfiguration-23trees}. 
Let $u$, $w$, and $z$ be the neighbors of $v$ in $G$. Note that $u$, $w$, and~$z$ bound a face $f$ of $G'$. Since $G$ is bi-acyclic and since the subgraph of $G$ induced by $v$, $u$, $w$, and $z$ is a complete graph, there exist a total order $O_\uparrow$ of these vertices in the upward direction, determined by $G^\uparrow$, and another total order~$O_\rightarrow$ in the rightward direction, determined by $G^\rightarrow$. Also, $v$ is neither the first nor the last vertex in these two orders, since it is an internal vertex of this subgraph. Analogously, there exist two total orders $O_\uparrow'$ and~$O_\rightarrow'$ of $u$, $w$, and $z$ in the upward and in the rightward direction, respectively, determined by $(G',Q')$; clearly, $O_\uparrow'$ and $O_\rightarrow'$ coincide with $O_\uparrow$ and~$O_\rightarrow$, respectively, when restricted to~$u$,~$w$, and $z$. Since $v$ lies in the interior of the $3$-cycle $(u,w,z)$ in the windrose-planar drawing of~$(G,Q)$ we used to construct $\mathcal{E}$, 
there exists a point in the interior of the triangle $uwz$ representing~$(u,w,z)$ in $\Gamma'$ such that placing $v$ on this point yields a straight-line windrose-planar drawing $\Gamma$ of $(G,Q)$. In fact, as already observed in the proof of Lemma~\ref{le:4-constrained-drawing}, there exist exactly two possible shapes of triangle $uwz$ in $\Gamma'$, which are determined by the two possible combinations of categories for the internal angles incident to $u$, $w$, and $z$. Namely, either one angle has category $0^\circ$, say $\angle{w,f}=0^\circ$, and $\angle{u,f}=\angle{z,f}=90^\circ$, or two angles have category $0^\circ$, say $\angle{w,f}=\angle{u,f}=0^\circ$, and $\angle{z,f}=180^\circ$. 
In both cases, placing $v$ in any point in the interior of triangle $uwz$ satisfies the q-constraints of any vertex with angle category equal to $0^\circ$, say angle $\angle{w,f}=0^\circ$, since in this case~$v$ belongs to the same quadrant of $w$ as $u$ and $z$, by Observation~\ref{obs:planarity}. 
Analogously, for the vertices with angle category equal to $90^\circ$, placing $v$ at any point in the interior of triangle $uwz$ trivially satisfies the constraints on the relative positions of $v$ with respect to $u$, $w$, and $z$ imposed by \emph{one} of $G^\uparrow$ and $G^\rightarrow$.

We show that it is always possible to find a point that satisfies the constraints imposed by \emph{both} $G^\uparrow$ and~$G^\rightarrow$.
Consider the case that $\angle{u,f}=\angle{z,f}=90^\circ$; refer to the setting depicted in Fig.~\ref{fig:triangle-3tree-a}, the other settings being symmetric. In this case, the vertical line passing through vertex $u$ and the horizontal line passing through vertex $z$ intersect at a point $p$ in the interior of triangle $uwz$. Thus, triangle $uwz$ intersects the four quadrants of point $p$ and we can hence place $v$ in one of these quadrants so to satisfy the q-constraints of $v$ with respect to $u$ and $z$.
Consider the case that $\angle{w,f}=\angle{u,f}=0^\circ$; refer to the setting depicted in Fig.~\ref{fig:triangle-3tree-b}, the other settings being symmetric. 
In this case, only three out of the four quadrants of $z$ intersect the interior of triangle $uwz$; assume, as in Fig.~\ref{fig:triangle-3tree-b}, that the \dl-quadrant does not intersect $uwz$. We claim that
$v \notin \Ndl(z)$. This is due to the fact that $w \in \Ndr(z)$, $u \in \Nul(z)$, and $v$ lies between $u$ and $w$ in the clockwise order of the neighbors of $z$. Then, the claim follows from Observation~\ref{obs:planarity}. Hence, we can place $v$ in one of these three quadrants so to satisfy the q-constraints of $v$ with respect to $z$.

The planarity of $\Gamma$ follows from the fact that placing a vertex in any interior point of a triangle and connecting it to its three vertices does not introduce any crossing.
\end{proof}

We now study the area requirements of straight-line windrose-planar drawings. Recall that, as already stated in Theorem~\ref{th:area-lowerbound}, the analogous result on upward planarity~\cite{dtt-arsdud-92} already implies an exponential lower bound on the area. In the following theorem we slightly improve this lower bound.

\begin{theorem}\label{th:exparea}
  There is a \qc graph  with $3k$ vertices such that any windrose-planar drawing requires area at least $4^{k-1}$ if the distance between the vertices is at least $1$.
\end{theorem}

\begin{proof}
  The graph consists of $k$ nested triangles whose edges are
  alternately directed in the NE-SW and in the NW-SE direction.
  Fig.~\ref{fig:big-area-1} shows two consecutive triangles $abc$ and
  $a'b'c'$ in this nested sequence.  The graph has additional edges
  connecting the triangles to ensure nesting, which are not shown.
  The area of the circumscribed box $R$ around $abc$ is at least twice
  the area of $abc$, since $abc$ can be extended to a parallelogram
  $abcd$ of double area which is still contained in $R$.
  \begin{figure}[tb]
  \centering
  \subfigure[Two successive triangles in the nested sequence.]{
  \includegraphics[page=1]{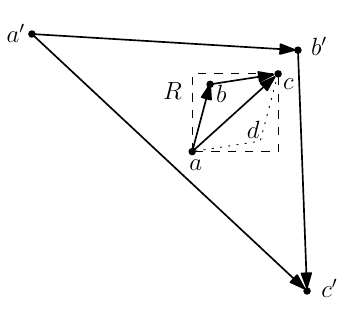}
  \label{fig:big-area-1}
  }
  \hfil
  \subfigure[The area of the triangle $a'b'c'$ is at least twice the area of 
    any enclosed rectangle~$R$.]{
  \includegraphics[page=2]{img/big-area}
  \label{fig:big-area-2}
  }
  \caption{A graph with exponential area requirement}
  \label{fig:big-area}
\end{figure}
  Moreover, as we show below, the area of the triangle $a'b'c'$ is at
  least twice the area of the rectangle~$R$. Finally, the area of the
  innermost triangle is at least $1/2$, and the theorem follows.
  
  To see that the area of $a'b'c'$ is at least twice the
  area of the enclosed rectangle~$R$, consider the smallest possible
  triangle for a fixed $R$; see Fig.~\ref{fig:big-area-2}. In this triangle, the edges~$(a',b')$ and~$(b',c')$ must be aligned with the edges of~$R$,
  and $(a',c')$ must touch $R$. A straightforward one-parameter
  minimization shows that the minimum area of $a'b'c'$ is achieved
  when $R$ touches the midpoint of~$(a',c')$, and then the area is twice
  the area of~$R$.
\end{proof}

\subsection{Straight-line bi-monotone drawings} We conclude the section by considering a problem, called {\sc Bi-monotonicity}~\cite{fpss-htmd-11}, that is % strictly ?closely?
 related to {\sc Windrose Planarity}. This problem takes as input a \emph{bi-ordered graph}, that is, a triple $\langle G(V,E), \gamma, \lambda \rangle$ where $G$ is a planar graph, while $\gamma: V \leftrightarrow \{1,\dots,n\}$ and $\lambda: V \leftrightarrow \{1,\dots,n\}$ are two bijective functions, each specifying a total order of $V$, and asks whether a \emph{bi-monotone drawing} of $\langle G(V,E), \gamma, \lambda \rangle$ exists, that is, a planar drawing of $G$ such that $x(u) < x(v)$ if and only if $\gamma(u) < \gamma(v)$, $y(u) < y(v)$ if and only if $\lambda(u) < \lambda(v)$, and edges are represented by $xy$-monotone curves. We say that a bi-ordered graph is \emph{bi-monotone} if it admits a bi-monotone drawing. 
In other words, while problem {\sc Windrose Planarity} asks to realize a \emph{partial} order among the vertices in one direction and another \emph{partial} order in the other direction, this problem asks to realize two \emph{total} orders. 
We prove that not all the bi-monotone graphs admit straight-line bi-monotone drawings; see Fig.~\ref{fig:unstretchable}.

\newcommand{\theorembimonotone}{There exists a bi-monotone bi-ordered graph $\langle G(V,E), \gamma, \lambda \rangle$ not admitting any straight-line bi-monotone drawing.}
\begin{theorem}\label{th:bimonotone}
\theorembimonotone
\end{theorem}

\begin{proof}
Graph $G(V,E)$ consists of cycles $(a,b,c,d,e)$ and $(a',b',c',d',e')$; 
function $\gamma$ induces a total order $e',c',a,b,d',d,b',a',c,e$, while function $\lambda$ induces a total order $e',b',c',d',a',a,d,c,b,e$.
A bi-monotone drawing of $\langle G(V,E), \gamma, \lambda \rangle$ is provided in Fig.~\ref{fig:unstretchable}. 

We only have to prove that $\langle G(V,E), \gamma, \lambda \rangle$ admits no straight-line bi-monotone drawing.
Namely, note that, in any straight-line bi-monotone drawing of $\langle G(V,E), \gamma, \lambda \rangle$, the lines through edges $(a,b)$ and $(c,d)$ must converge on the side of $b$ and $c$, that is, the half-line starting at $a$ and passing through $b$ must intersect the half-line starting at $d$ and passing through~$c$, as otherwise edges $(a,e)$ and $(d,e)$ could not be drawn as straight-line segments. Hence, 
$\Delta_y(a,b)/\Delta_x(a,b) < \Delta_y(c,d)/\Delta_x(c,d)$ must hold, where $\Delta_\circ(\alpha,\beta)=|\circ(\beta)-\circ(\alpha)|$ with $\circ \in \{x,y\}$. However, since 
$\Delta_y(a,b) > \Delta_y(c,d)$ holds due to function $\lambda$, in order to satisfy the above inequality it must be $\Delta_x(a,b) > \Delta_x(c,d)$. With a symmetrical argument on cycle $(a',b',c',d',e')$, one can prove that it must be $\Delta_x(a',b')> \Delta_x(c',d')$. However, function $\gamma$ enforces $\Delta_x(a,b) < \Delta_x(c',d')$ and $\Delta_x(a',b') < \Delta_x(c,d)$. Hence, $\Delta_x(a,b) < \Delta_x(c',d') < \Delta_x(a',b') <  \Delta_x(c,d) < \Delta_x(a,b)$, and the statement follows.
\end{proof}
\begin{figure}[tb]
  \centering
  \includegraphics[width=0.3\textwidth]{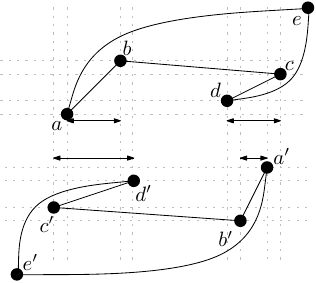}
  \caption{A bi-monotone bi-ordered graph that does not admit any straight-line bi-monotone drawing.}\label{fig:unstretchable}
\end{figure}

\section{Conclusions and open problems}\label{se:conclusions}

 We have studied the concept of windrose planarity of a graph, that is planarity where each neighbor of a vertex $v$ is constrained to lie in a specific quadrant of $v$. Besides its theoretical appeal and its practical applications, we studied this new notion of planarity because of its methodological relevance. Namely, graphs have been studied for centuries focusing both on their abstract topological nature and on their geometric representations. In this dichotomy, windrose planarity essentially has an intermediate position, since a windrose-planar graph, while still being an abstract topological structure, has already absorbed -- because of the relative positions among its adjacent vertices -- a fair amount of geometry. 

We have shown that if the combinatorial embedding of a graph is fixed, then windrose planarity can be tested in polynomial time. We also gave several contributions on the question whether a windrose-planar graph admits a straight-line (windrose-planar) drawing, which is probably the most studied geometric representation of graphs.

Several interesting problems arise. 
\begin{inparaenum}[(i)]
\item Does a windrose-planar graph always admit a straight-line windrose-planar drawing? The usual methods for constructing planar straight-line drawings~\cite{fpp-hdpgg-90,s-epgg-90} do not seem to be easily extended to cope with this.
\item In Section~\ref{se:straight-line} we have stated that every windrose-planar \qc graph whose blocks are either edges or planar $3$-trees admits a straight-line windrose-planar drawing, however our techniques might produce drawings whose vertices are placed arbitrarily close to each other. Are there algorithms for this family of graphs that, assuming a finite resolution rule, produce drawings with polynomial area?
\item  The constraints on the relative positions of the adjacent
  vertices can be relaxed to two adjacent quadrants. 
For example, for a vertex $u$ with neighbors $v$, $w$, and $z$ in counterclockwise order, one can specify that $v$ is either NE or NW of $u$, $w$ is SW, and $z$ is either NE or SW.
Is this problem still polynomial? 
We remark that this version of the problem allows to simultaneously visualize two partial orders defined by means of different edge sets, provided that
their union is planar. One would then color the edges to indicate       whether they belong to one or the other poset or both.
\end{inparaenum}

\bibliographystyle{abbrvurl}
\bibliography{bibliography}

\end{document}